\documentclass{article}

\usepackage{color}
\usepackage{amsmath}
\usepackage{amssymb}
\usepackage{graphicx}
\usepackage{wasysym}
\usepackage{psfrag}
\usepackage{authblk}

\newcommand{\myqed}{\hspace*{\fill}$\blacksquare$}

 \newtheorem{thm}{Theorem}
 \newtheorem{lem}[thm]{Lemma}
 \newtheorem{cor}[thm]{Corolary}
 \newtheorem{defn}[thm]{Definition}
 \newtheorem{prop}[thm]{Proposition}
 \newenvironment{sep}[1]{\vspace{5pt}\noindent {\bf #1 ---}}{}
 \newenvironment{proof}{\noindent {\bf Proof ---}}{\myqed}

\newcommand{\fref}[1]{Fig.~{\ref{fig:#1}}}
\newcommand{\propref}[1]{proposition~{\ref{prop:#1}}}
\newcommand{\lemref}[1]{lemma~{\ref{lem:#1}}}
\newcommand{\corref}[1]{corollary~{\ref{cor:#1}}}
\newcommand{\thmref}[1]{theorem~{\ref{thm:#1}}}
\newcommand{\secref}[1]{section~{\ref{sec:#1}}}
\newcommand{\Secref}[1]{Section~{\ref{sec:#1}}}

\newcommand{\vect}[1]{\boldsymbol{\mathrm{#1}}}

\newcommand{\be}{\begin{equation}}
\newcommand{\ee}{\end{equation}}
\newcommand{\bex}{\begin{equation*}}
\newcommand{\eex}{\end{equation*}}
\newcommand{\nl}{\nonumber\\}

\newcommand{\Z}{\mathbf{Z}}
\newcommand{\N}{\mathbf{N}}
\newcommand{\C}{\mathbf{C}}
\newcommand{\id}{\mathbf{1}}

\newcommand{\st}{\,:\,}
\newcommand{\sset}[1]{\{#1\}}
\newcommand{\set}[2]{\{#1\st #2\}}
\newcommand{\sget}[1]{\langle#1\rangle}

\newcommand{\funct}[3]{#1:\,#2\longrightarrow #3}
\newcommand{\functmap}[5]
{
\begin{alignat}{1}
#1\,:\,\,#2 &\longrightarrow #3\nl
#4&\longmapsto #5
\end{alignat}
}
\newcommand{\functmapdouble}[9]
{
\begin{alignat}{3}
#1\,:\, #2 &\longrightarrow #3 \qquad &#1\,:\, #6 &\longrightarrow #7\nl
#4&\longmapsto #5\qquad &#8&\longmapsto #9
\end{alignat}
}

\newcommand{\inv}{^{-1}}

\newcommand{\powerset}[1]{{\bf P}(#1)}
\newcommand{\powersetfin}[1]{{\bf P}_{0}(#1)}

\newcommand{\pauli}{\mathcal{P}}
\newcommand{\gauge}{\mathcal{G}}
\newcommand{\stab}{\mathcal{S}}

\newcommand{\cen}[2]{\mathcal{Z}_{#2}({#1})}

\newcommand{\ket}[1]{|#1\rangle}

\newcommand{\cnstrx}{\mathrm{Cnstr}}
\newcommand{\cnstr}[1]{\cnstrx (#1)}
\newcommand{\cnstrr}[2]{\cnstrx_{#1}(#2)}
\newcommand{\cnstrfinx}{\cnstrx^0}
\newcommand{\cnstrfin}[1]{\cnstrfinx(#1)}
\newcommand{\cnstrfinr}[2]{\cnstrfinx_{#1}(#2)}

\newcommand{\suppx}{\mathrm{Supp}}
\newcommand{\supp}[1]{\suppx(#1)}
\newcommand{\suppi}[1]{\suppx[#1]}
\newcommand{\suppr}[2]{\suppx_{#1}(#2)}

\newcommand{\negrx}[1]{\mathrm {Neg}_{#1}}
\newcommand{\negr}[2]{\negrx{#1}(#2)}
\newcommand{\chgx}{\mathrm {Chg}}
\newcommand{\chgrx}[1]{\chgx_{#1}}
\newcommand{\chg}[1]{\chgx(#1)}
\newcommand{\chgr}[2]{\chgrx{#1}(#2)}
\newcommand{\syndx}{\mathrm {Com}}
\newcommand{\syndrx}[1]{\syndx_{#1}}

\newcommand{\syndr}[2]{\syndx_{#1}(#2)}
\newcommand{\syndrset}[2]{\syndx_{#1}(#2)}
\newcommand{\syndri}[2]{\syndx_{#1}[#2]}

\newcommand{\prox}{\mathrm{Pro}}
\newcommand{\pro}[1]{\prox(#1)}
\newcommand{\proi}[1]{\prox[#1]}
\newcommand{\morx}[1]{\negrx{#1}^{-1}}
\newcommand{\mor}[2]{\morx{#1}(#2)}

\newcommand{\comm}[2]{(#1;#2)}
\newcommand{\bl}[4]{\square_{#1,#2}^{#3,#4}}
\newcommand{\site}[2]{(#1,#2)}
\newcommand{\thk}{\mathrm {Thk}}
\newcommand{\trans}[2]{T_{#1,#2}\,}
\newcommand{\strans}[1]{T^{(#1)}\,}
\newcommand{\overlap}[1]{||_{#1}}
\newcommand{\range}[1]{\rho(#1)}
\newcommand{\coarsex}[1]{\mathrm{Crs}_{\hspace{.3pt}#1}}
\newcommand{\coarse}[2]{\coarsex{#1}(#2)}

\newcommand{\morphism}[3]{#2\,\overset{#1}\longrightarrow\,#3}
\newcommand{\morphismd}[5]{#3\,\overset{#1}\longrightarrow\,#4\,\overset{#2}\longrightarrow\,#5}

\newcommand{\PA}{\mathcal{A}}
\newcommand{\PB}{\mathcal{B}}

\newcommand{\phase}{\sget{i\id}}
\newcommand{\morph}[1]{\Phi(#1)}
\newcommand{\morphfin}[1]{\Phi^0(#1)}

\newcommand{\charge}[1]{C_{#1}}
\newcommand{\spinx}{\theta}
\newcommand{\spin}[1]{\spinx(#1)}
\newcommand{\mutualx}{\kappa}
\newcommand{\mutual}[2]{\mutualx(#1,#2)}
\newcommand{\stabtrivial}{\stab_{\mathrm{T}}}
\newcommand{\stabtrivialg}{\stab^{\mathrm{T}}_g}
\newcommand{\stabstrivial}{\stab_{\mathrm{S}}}

\newcommand{\gaugestrivial}{\gauge_{\mathrm{S}}}

\newcommand{\idtrivial}{1_{\mathrm{T}}}

\newcommand{\paulitrivial}{\pauli_{\mathrm{T}}}
\newcommand{\stabTC}{\stab_{\mathrm{TC}}}
\newcommand{\stabTCg}{\stab^{\mathrm{TC}}_g}
\newcommand{\pauliTC}{\pauli_{\mathrm{TC}}}
\newcommand{\stabSTC}{\stab_{\mathrm{STC}}}
\newcommand{\gaugeSTC}{\gauge_{\mathrm{STC}}}
\newcommand{\stabSTCg}{\stab^{\mathrm{STC}}_g}
\newcommand{\gaugeSTCg}{\gauge^{\mathrm{STC}}_g}
\newcommand{\stabFTC}{\stab_{\mathrm{FTC}}}
\newcommand{\gaugeFTC}{\gauge_{\mathrm{FTC}}}
\newcommand{\stabFTCg}{\stab^{\mathrm{FTC}}_g}
\newcommand{\gaugeFTCg}{\gauge^{\mathrm{FTC}}_g}

\newcommand{\stabhon}{\stab_{\hexagon}}
\newcommand{\gaugehon}{\gauge_{\hexagon}}
\newcommand{\stabhong}{\stab^{\hexagon}_g}
\newcommand{\gaugehong}{\gauge^{\hexagon}_g}

\newcommand{\stabTSCC}{\stab_{\mathrm{col}}}
\newcommand{\gaugeTSCC}{\gauge_{\mathrm{col}}}
\newcommand{\stabTSCCg}{\stab^{\mathrm{col}}_g}
\newcommand{\gaugeTSCCg}{\gauge^{\mathrm{col}}_g}

\newcommand{\pth}[1]{\mathrm{Path}(#1)}
\newcommand{\strx}{\mathrm{Str}}
\newcommand{\str}[2]{\strx(#1;#2)}
\newcommand{\strp}[3]{\str{#1}{\pth{#2,#3}}}
\newcommand{\strd}[2]{\strx(#1,#2)}

\newcommand{\Site}{{\mathbf{0}}}
\newcommand{\graph}{\Gamma}
\newcommand{\dgraph}{\Gamma^\ast}

\newcommand{\edges}{{\Gamma_\mathrm{edg}}}
\newcommand{\dedges}{{\Gamma_\mathrm{edg}^\ast}}
\newcommand{\faces}{{\Gamma_\mathrm{fc}}}
\newcommand{\dfaces}{{\Gamma_\mathrm{fc}^\ast}}

\newcommand{\hedge}{\vect e_1}
\newcommand{\vedge}{\vect e_2}
\newcommand{\face}{\vect f}
\newcommand{\dface}{\vect f^\ast}

\newcommand{\edgeopx}[1]{\epsilon_{#1}}
\newcommand{\hedgeopx}[1]{\hat\epsilon_{#1}}
\newcommand{\edgeopxp}[1]{\epsilon'_{#1}}
\newcommand{\edgeop}[2]{\edgeopx{#1}(#2)}
\newcommand{\hedgeop}[2]{\hedgeopx{#1}(#2)}
\newcommand{\dedgeop}[2]{\edgeopx{#1}(#2^\ast)}
\newcommand{\edgeopp}[2]{\edgeopxp{#1}(#2)}
\newcommand{\edgeopi}[2]{\edgeopx{#1}[#2]}

\newcommand{\indall}{K}
\newcommand{\indc}{K_c}
\newcommand{\indd}{K_d}
\newcommand{\inde}{K_e}

\newcommand{\mybox}[1]{
\begin{center}
\noindent\emph{#1}
\end{center}
}

\newcommand{\add}{+}

\begin{document}

\date{}
\title{Structure of 2D Topological Stabilizer Codes}
\author{H\'ector Bomb\'in}
\affil{Perimeter Institute for Theoretical Physics\\ 31 Caroline St. N., Waterloo, ON, N2L 2Y5, Canada}

\maketitle

\begin{abstract}
We provide a detailed study of the general structure of translationally invariant two-dimensional topological stabilizer quantum error correcting codes, including subsystem codes.
We show that they can be understood in terms of the homology of string operators that carry a certain topological charge.
In subsystem codes two dual kinds of charges appear.
We prove that two non-chiral codes are equivalent under local transformations iff if they have isomorphic topological charges.
Our approach emphasizes local properties over global ones.
\end{abstract}

\newpage
\tableofcontents

\section{Introduction}

Quantum error correcting codes \cite{shor:1995:scheme,steane:1996:error,knill:1997:theory,bennett:1996:mixed} play a fundamental role in the quest to overcome the decoherence of quantum systems.
Among them, stabilizer codes \cite{gottesman:1996:stabilizer,calderbank:1997:quantum} provide a large and flexible class of codes that are at the same time relatively easy to investigate.
Interactions between physical qubits are sometimes subject to locality constraints in a geometrical sense. 
E.g., qubits might be placed in a $D$-dimensional array with interactions available only between nearest neighbors.
In such cases the subclass of topological stabilizer codes (TSC) \cite{kitaev:2003:ftanyons,dennis:2002:tqm,bombin:2006:2dcc,bombin:2007:3dcc,bombin:2010:subsystem,bravyi:2010:topological,suchara:2011:constructions,haah:2011:local} is a natural choice.
TSCs not only have very nice locality properties, but are also flexible in terms of the manipulation strategies that they allow. 

The purpose of this paper is to investigate the general structure of two-dimensional TSCs,
covering both subspace codes and the more general subsystem codes \cite{kribs:2005:unified,poulin:2005:stabilizer,bacon:2006:operator}.
The only constraint imposed is translational symmetry of the bulk of the code, which provides a way to list the codes.

A stabilizer quantum error correcting code is described in terms of certain `check operators', a set of commuting observables that have to be measured in order to get information about what errors have affected an encoded state.
In TSCs these measurements are local in a geometrical sense.
In particular, a TSC is in practice usually given as a recipe to construct the check operators on a given lattice of qubits.
The lattice can be arbitrarily large, but check operators are defined locally, with support on a set of qubits contained in a bounded region.
What makes these codes topological, as opposed to simply local, is that no information about the encoded qubits can be recovered if only access to a local set of qubits is granted.
Indeed, operators on encoded qubits have support on a number of physical qubits that grows with the lattice size.

An essential feature of known topological codes is that they have an error threshold \cite{dennis:2002:tqm,wang:2003:confinement,katzgraber:2009:cc,katzgraber:2009:unionjack,andrist:2010:tricolored}.
I.e., in the limit of large lattices, for noise below a certain threshold error correction is asymptotically perfect.
This is true either in a simple error correction scenario or in a fault-tolerant scenario, and also when qubit losses are taken into account \cite{stace:2010:error}.
It is also often true \cite{yoshida:2010:classification} ---but for interesting exceptions, see\cite{bravyi:2010:topological,haah:2011:local}--- that the number of encoded qubits depends only on the homology of the qubit lattice.
In particular a trivial homology gives no encoded qubits, but even in a planar setting it is possible to recover non-trivial codes by introducing suitable boundaries  \cite{freedman:2001:planar} and other defects such as twists \cite{bombin:2010:twist}.
Moreover, because in such `homological' codes the lattices can be chosen very flexibly, it is possible to carry out certain computations by changing the code geometry over time, something called `code deformation' \cite{dennis:2002:tqm,raussendorf:2007:deformation,bombin:2009:deformation}.
Error correction turns out to be connected to classical statistical models \cite{dennis:2002:tqm,katzgraber:2009:cc,bombin:2010:subsystem} and this has given rise to fast algorithms to infer the errors \cite{duclos:2010:fast}.
Finally, there exist topological codes that are especially well suited to perform computations using transversal gates \cite{bombin:2006:2dcc,bombin:2007:3dcc}, which minimizes error propagation.

General properties of TSCs in two or higher dimensions have already been explored to different degrees. 
General constraints on the code distance ---directly related to the geometry of encoded operators--- were first found in \cite{bravyi:2009:no-go}, and improved in \cite{bravyi:2010:tradeoffs}.
The subsystem case was addresed in \cite{bravyi:2011:subsystem}.
The geometry of logical operators in a subclass of TSCs, subject to constraints such as scale invariance of the number of encoded qubits, was studied in \cite{yoshida:2010:classification}.
Constraints on the code distance for three-dimensional codes that do not satisfy this scale invariance condition have been recently developed \cite{bravyi:2011:energy}.

One of the main results in this paper is that all two-dimensional TSCs can be understood in terms of `string operators' that carry a `topological charge'.
In particular, homology plays an essential role and dictates the number of encoded qubits.
In the case of subspace codes, the corresponding Hamiltonian model \cite{kitaev:2003:ftanyons} ---which has the code as its gapped ground state--- exhibits anyonic excitations.
Moreover, all subspace codes that give rise to the same anyon model turn out to be equivalent, in the sense that there exists a local transformation connecting them.
The implications of this result, both from the condensed matter and quantum information perspectives, are explored in \cite{bombin:2011:universal}. 
In the case of subsystem codes we no longer have a direct interpretation in terms of physical anyons.
However, we find a nice duality structure between two kinds of charges.
The first kind corresponds to an anyon model ---possibly chiral---, and the second to fluxes with which the anyons interact topologically.
We show that codes giving rise to the same non-chiral anyon model are locally equivalent.
Computational implications of these results are discussed in the conclusions.

Our approach emphasizes local properties of the codes ---such as the structure of check operators--- over global ones ---such as the number of encoded qubits.
In practice we realize this by considering infinite versions of the codes that cover the plane.
The idea is that all operators acting on a finite number of qubits become automatically local in this infinite picture.
At the same time, the homology of the plane is trivial and this simplifies the analysis.

The paper is divided as follows.
\Secref{results} informally summarizes the main results, providing the intuition behind the main definitions and proofs.
\Secref{basics} introduces in a formal manner topological stabilizer groups (TSGs), used to model TSCs.
Subsequent sections develop diverse aspects of the structure of TSGs, culminating in a structure theorem in \secref{classification} that we use to prove local equivalence. 
Finally, \secref{conclusions} discusses natural extensions of this work.

To avoid repetition, we will often omit the qualifier ``two-dimensional'' when discussing lattices, TSCs, and so on, but it should be understood in all cases.

\section{Approach and results}\label{sec:results}

The purpose of this section is to summarize informally the approach taken to investigate TSCs and the results obtained.
At the same time, it explains the motivation behind the main definitions and the intuition behind some proofs.

\subsection{Stabilizer codes}

Given a system with $n$ qubits, its Pauli group is
\be\label{pauli}
\pauli_n:=\sget{i\id, X_1,Z_1,\dots,X_n,Z_n},
\ee
where $X_i$, $Z_i$ denote the Pauli $X$ and $Z$ operators acting on the $i$-th qubit.
A stabilizer code \cite{gottesman:1996:stabilizer,calderbank:1997:quantum} on $n$ qubits is defined by a subgroup of Pauli operators $\stab\subset\pauli$, called stabilizer group, such that $-1\not \in\stab.$
$\stab$ is then abelian and its elements are self-adjoint.
The code is composed by those states $\ket \psi\in\C^{2^n}$ such that
$
s \ket\psi =\ket\psi
$
for every stabilizer $s\in\stab$.
It is enough to check these conditions for a set $\sset {s_i}$ of independent generators of $\stab$.
Such stabilizer generators serve as check operators: they can be measured in order to recover information about any errors affecting the encoded states.
If $\stab$ has $n-k$ independent generators the code subspace has dimension $2^k$, it encodes $k$ logical qubits.
The logical or encoded Pauli operators are recovered as the quotient
\be
\frac{\cen\stab{\pauli_n}}\stab \simeq \pauli_{k},
\ee
where $\cen\stab{\pauli_n}$ denotes the centralizer of $\stab$ in $\pauli_n$.
This centralizer is the group of undetectable Pauli errors: those that do not change the error syndrome.
The quotient is necessary because the elements of $\stab$ are trivial undetectable errors, without any effect on encoded states.

Given a set $\sset{s_i}$ of generators of $\stab$ we can write down a Hamiltonian
$
H=-\sum_i s_i.
$
Its ground subspace is the code subspace, and there is a gap of two energy units to any excited state.
This gap is important when we consider families of TSCs with local generators in systems of arbitrary size, because the size independence of the gap gives rise to a gapped phase.
An error syndrome in the code amounts to an excitation configuration in the Hamiltonian system.

It is also possible to encode only in a subsystem of the code subspace.
In the stabilizer formalism this logical subsystem can be described as that in which the action of a certain gauge group $\gauge\subset\pauli$ is trivial \cite{poulin:2005:stabilizer}.
The gauge group satisfies
\be\label{gauge_stab}
\cen\gauge\gauge \propto \stab
\ee
and we assume $\phase\subset\gauge$.
Logical operators are recovered from
\be
\frac{\cen\gauge{\pauli_n}}\stab \simeq \pauli_{k},
\ee
where $2^k$ is no longer the dimension of the code subspace, which is instead $2^{k+r}$ for some $r$.
Then $\stab$ has $n-k-r$ independent generators and $\gauge$ has $n-k+r$ independent generators.
Undetectable errors form the group $\cen\stab{\pauli_n}$.
They are, up to a product with an element of $\gauge$, logical operators.
Indeed,
\be
\frac{\cen\gauge{\pauli_n}}{\phase\stab} \simeq \frac{\cen\stab{\pauli_n}}{\gauge}.
\ee
Stabilizer subsystem codes can be described just by giving $\gauge$, since then $\stab$ is fixed up to signs that can be chosen arbitrarily.

For subsystem codes we no longer have such a straightforward interpretation in terms of a quantum Hamiltonian.
A natural possibility is to take a set of generators $\sset{g_i}$ of the gauge code as the Hamiltonian terms.
This will guarantee that all the energy eigenvalues of the system have at least a degeneracy $2^k$, since encoded operators commute with all Hamiltonian terms.
However, there is no reason for an energy gap to persist in the limit of large system sizes.

\subsection{Topological codes and lattice Pauli groups}

For a TSC we mean a family of codes such that, loosely speaking, stabilizer generators are local, and non-trivial undetectable errors are non-local.
In addition, for subsystem codes we impose that gauge generators are local.
As we will see, certain topological subsystem codes naturally give rise to a few global generators either in the stabilizer or in the gauge group, but in such a way that it is irrelevant from the encoding perspective.
There exist codes with local gauge generators and non-local stabilizer generators, such as Bacon-Shor codes \cite{bacon:2006:operator},
but we do not consider them topological.
Indeed, they lack an essential feature of topological codes: a non-zero error threshold in the limit of large code size.

In order to characterize TSCs in a more concrete manner, let us first consider ordinary stabilizer codes, not subsystem ones.
As discussed in the introduction, a TSC is usually given as a recipe to construct the check operators on a given lattice of qubits.
We assume that this recipe, in the bulk of the code ---e.g., ignoring possible boundaries or other defects--- takes the form of a local and translationally invariant (LTI) set of check operators.
Then we can consider an infinite version of the same lattice and construct using the recipe an infinite group with check operators as generators.
We call this a lattice Pauli group (LPG).
Translational invariance is supposed to hold only at a suitable scale.

In general, for a LPG we mean a group that has as generators a LTI set of Pauli operators on a certain lattice of qubits. 
By definition, the Pauli group $\pauli$ on the infinite lattice of qubits is itself a LPG, with generators all the single qubit Pauli operators.
The goal is to understand which properties distinguish those LPGs that correspond to a TSC, which we call topological stabilizer groups (TSG).
Then we can shift our study of TSCs to that of TSGs.

We define TSGs as LPGs $\stab$ satisfying $-1\not\in\stab$ and the topological condition
\be\label{topo_condition}
\cen\stab\pauli\propto\stab.
\ee
Since in the infinite lattice all the Pauli operators are local, this condition reads:
\mybox{Local undetectable errors do not affect encoded states.}
To motivate \eqref{topo_condition} in detail, let $\stab$ be the LPG corresponding to a certain TSC given by the family of stabilizer groups $\stab_l$.
Any operator $O\in\cen\stab\pauli$ has an analog $O'$ acting on the bulk of the $l$-th code in the family for $l$ sufficiently large.
Moreover, $O'\in\cen{\stab_l}\pauli$ due to the locality of the stabilizer generators, and no matter what detailed definition of locality we adopt, $O'$ should be local for sufficiently large $l$.
Since $O'$ is undetectable and the codes $\stab_l$ are topological, it follows that $\phi O'\in\stab_l$ for some phase $\phi$.
Then if $\phi O\not\in \stab$ we can safely add $\phi O$, and all its translations, to an LTI generating set of $\stab$, enlarging $\stab$.
The only question left then is whether one might have to keep adding larger and large generators, but this is not the case because $\cen\stab\pauli$ has a LTI set of generators.
This is guaranteed by \corref{centralizer}:
\mybox{The centralizer of a LPG is a LPG.}

Now consider the subsystem case.
We define  topological stabilizer subsystem groups (TSSGs) as LPGs $\stab$ satisfying $-1\not\in\stab$ and the topological condition
\be\label{topo_condition_gauge}
\cen{\cen\stab\pauli}\pauli\propto\stab.
\ee
This condition is trivially satisfied when the number of qubits is finite, but is not generally true for LPGs, as an example in the next section shows.
To motivate \eqref{topo_condition_gauge}, let $\stab$ be the LPG corresponding to a certain topological stabilizer subsystem code given by the family of stabilizer groups $\stab_l$ and gauge groups $\gauge_l$.
We first notice that it makes sense to define the gauge group as the LPG
\be\label{gauge}
\gauge:=\cen\stab\pauli.
\ee
Indeed, any operator $O\in\cen\stab\pauli$ has a local analog $O'\in\cen{\stab_l}\pauli$ acting on the bulk of the $l$-th code in the family for $l$ sufficiently large.
Since $O'$ is undetectable and the codes $(\stab_l,\gauge_l)$ are topological, it follows that $\phi O'\in\gauge_l$ for some phase $\phi$ ---but phases are unimportant in the gauge group.
Conversely, any operator $O'$ of $\gauge_l$ with support in the bulk corresponds to an operator $O\in\cen\stab\pauli$ due to the locality of the stabilizer generators.
Second, due to the definition \eqref{gauge} we have $\stab\subset\gauge$ and from that $\cen\gauge\pauli\subset \gauge$ and $\cen\gauge\pauli=\cen\gauge\gauge$.
Any operator $O\in\cen\gauge\gauge$ has an analog $O'\in\cen{\gauge_l}{\gauge_l}$ for $l$ sufficiently large, and thus we can if needed enlarge $\stab$ as above to get $\stab\propto\cen\gauge\pauli$, which does not affect \eqref{gauge}.

\subsection{Examples}

The most important example of TSG is the \emph{toric code}.
In an infinite square lattice, we place one qubit per link and define two kinds of operators. 
For each vertex $v$ there is a vertex operator $X_v=X_1X_2X_3X_4$ where $i=1,2,3,4$ identify the qubits lying on those links meeting at $v$. Similarly, for each face $f$ there is a face operator $Z_f=Z_1Z_2Z_3Z_4$ with $i=1,2,3,4$ now identifying the qubits lying at those links composing the boundary of $f$. The stabilizer $\stab$ is generated by all such vertex and face operators. Using homology theory it is not difficult to check that it satisfies \eqref{topo_condition}.

To illustrate TSSGs, we can consider a variation of the toric code that, although useless as a quantum code, has nontrivial structure as we will see.
For this \emph{subsystem toric code} $\stab$ is generated by face operators $Z_f$. 
It follows from the properties of the toric code that $\gauge$ has as generators vertex operators $X_v$ and single qubit $Z$ Pauli operators.
Moreover, \eqref{topo_condition_gauge} is satisfied.

\subsection{Independent generators and constraints}

Given a LTI set $G$ of generators of a LPG, a constraint is a subset of $G$ such that the product of its elements is proportional to the identity ---but none of the elements is.
We distinguish between local constraints, with a finite number of elements, and global constraints, with a possibly infinite number of elements.
The latter are well defined thanks to locality.
Elements of $G$ are independent if they are not part of a local constraint.

LPGs $\PA$ that are the centralizer of some other LPG $\PB$, so that $\PA=\cen\PB\pauli$, have very specific properties regarding independent generators and constraints:
\mybox{The centralizer of a LPG admits a LTI set of independent generators. Such a set of generators has a finite number of global constraints.}
These results are \lemref{global_constraints} and \thmref{independent}, but see \secref{centralizer}.

For general LPGs it is not true that there exists always a LTI set of independent generators.
As a counterexample ---that also shows that \eqref{topo_condition_gauge} is not satisfied for any abelian LPG---, one can consider a square lattice with one qubit per site, and the LPG generated by nearest neighbor Pauli operators of the form $Z_iZ_j$. 
There is one generator per link, and they are not independent because the product of all the link operators in a closed circuit is the identity operator.
Moreover, it is not difficult to see that it is not possible to choose a translationally invariant set of independent generators.
Similarly, the result does not remain true for higher dimensions.
A counterexample in dimension three is given by the 3D toric code \cite{dennis:2002:tqm}, where the product of the plaquette operators forming a closed surface is the identity operator.

In both of the examples offered above, the toric code and the subsystem toric code, the given generators for $\stab$ and $\gauge$ are independent, as one can easily check. 
The set of all face operators gives rise to a global constraint, and the same is true for the set of all vertex operators. 
Thus, there are in total four global constraints, including the trivial empty set.
In the subsystem case both the generators of $\stab$ and $\gauge$ give rise to the same number of global constraints, two.
This is always so.

\subsection{Charge}

A central notion in this paper is that of charge.
To introduce it, first consider the commutator of two Pauli elements
\be
\comm p q:= pqp\inv q\inv=\pm 1.
\ee
Given $p\in\pauli$ we can construct a group morphism $\funct {\comm p{\,\cdot\,}}{\pauli}{\pm 1}.$
Obviously $\phase$ is a subset of the kernel of $\comm p{\,\cdot\,}$.
So let us consider more generally, given a LPG $\PA$, the set of morphisms $\morph\PA$ from $\PA$ to $\pm 1$ that satisfy this property.
They form an abelian group with product $(\phi\phi')(a):=\phi(a)\phi'(a)$.
We are especially interested in the subgroup $\morphfin\PA\subset\morph\PA$ of those morphisms such that the preimage of $-1$ contains a finite number of elements of any LTI set of generators.
E.g., if $\PA$ is a TSG, then the elements of $\morphfin\PA$ represent states of the corresponding Hamiltonian model with a finite number of excitations.
Since $\comm p{\,\cdot\,}\comm q{\,\cdot\,}=\comm {pq}{\,\cdot\,}$, the elements of $\morphfin\PA$ of the form $\comm p\cdot$ form a subgroup, that we denote $\syndri\PA\pauli$.
We are interested in those elements of $\morphfin\PA$ that do not correspond to a Pauli operator.
This leads to define the charge group 
\be
\charge\PA :=\frac {\morphfin\PA}{\syndri\PA\pauli}.
\ee
The original motivation for this definition comes from the case in which $\PA$ is a TSG.
In that case the equivalence classes can be regarded as classes of excitation configurations up to local transformations.
I.e., charge is conserved in any given region if the only transformations allowed are those that affect only that region.
For a centralizer LPG the charge group is finite because it is dual to the group of global constraints.
Interestingly, in the case of TSSGs it is useful to consider both the stabilizer charge group $\charge\stab$ and the gauge charge group $\charge\gauge$, which turn out to be dual too.

Consider the toric code example. 
The $Z_e$ single qubit Pauli operator corresponding to a given link $e$ anticommutes with the two vertex operators corresponding to the vertices $v, v'$ on the ends of the link, and commutes with all other vertex and face operators.
Moreover, given any even set of vertices $v_1,\dots,v_{2n}$ we can find a set of links $E$ such that $\prod_E Z_e$ anticommutes only the vertex operators $X_{v_i}$, $i=1,\dots,2n$ ---indeed, consider any set of links forming $n$ `strings', with the $j$-th string linking the $2j-1$ and $2j$ vertices. 
The reasoning holds for face operators up to duality in the lattice.
Therefore, the elements of $\syndri\stab\pauli$ are exactly those elements of $\morphfin\stab$ with value $-1$ on an even number of vertex operators and an even number of face operators. 
There are then four charges in $\charge\stab$, corresponding to the two possible values of vertex and face parities, in agreement with the number of global constraints.

In the subsystem toric code we only have face generators in $\stab$, so that the number of charges in $\charge\stab$ is two as expected. 
As for $\charge\gauge$, $Z_e$ pauli operators can be individually `flipped' with $X_e$ without affecting other generators.
Thus, $\charge\gauge$ has two elements corresponding to the parity of vertex operators with value $-1$.

\subsection{Strings and topological charge}

\begin{figure}
\begin{center}
\includegraphics[width=.45\columnwidth]{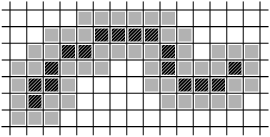}
\end{center}
\caption{
We represent sites as the faces of a square lattice.
The striped sites form a path.
Together with the grey sites they form the support of a corresponding string operator.}
\label{fig:string}
\end{figure}

\begin{figure}
\psfrag{(a)}{(a)}
\psfrag{(b)}{(b)}
\psfrag{(c)}{(c)}
\psfrag{(d)}{(d)}
\psfrag{p}{$p$}
\psfrag{pp}{$p'$}
\psfrag{q}{$q$}
\psfrag{r}{$r$}
\psfrag{qp}{$q'$}
\psfrag{rp}{$r'$}
\psfrag{u}{$u$}
\psfrag{v}{$v$}
\psfrag{w}{$w$}
\psfrag{x}{$x$}
\begin{center}
\includegraphics[width=.9\columnwidth]{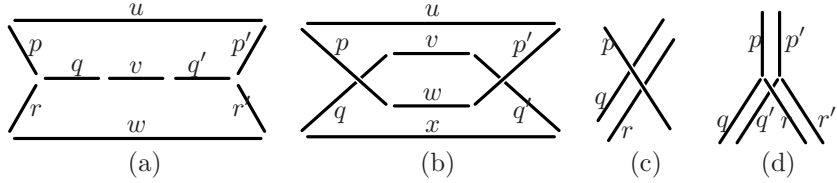}
\end{center}
\caption{
Geometries involved in the definition of string commutation rules.
}
\label{fig:commutation}
\end{figure}

Consider the charge group $\charge\gauge$ of a TSSG.
Given a LTI set of generators of $\gauge$, let us say that the support of $\phi\in\morph\gauge$ is the set of lattice sites for which there exist a generator $g$ with support in the site and $\phi(g)=-1$.
Since the number of charges is finite, we can coarse grain the lattice till each site can hold any of the charges, in the sense that there exists for any site and for any charge $c\in\charge\gauge$ a morphism $\phi\in\morphfin\gauge$ with charge $c$ and support at that site.
Consider two such morphisms $\phi,\phi'$ with the same charge $c$ but support at different sites $\sigma$, $\sigma'$.
Their product has trivial charge so that $\phi\phi'=\comm p {\,\cdot\,}$ for some $p\in\pauli$.
Moreover, if we coarse grain enough we can choose $p$ to have support along a string-shaped region ---a thickened path, see \fref{string}--- connecting $\sigma$ and $\sigma'$.
This is a matter of choosing a suitable set of operators $p_i$ for adjacent sites ---one per valid $\phi\phi'$ pair---, coarse graining till all the $p_i$ fit in sites adjacent to those involved, and finally forming string operators by composing these $p_i$ operators.
We say that $p$ is a string operator with charge $c$ and endpoints $\phi$, $\phi'$.
In the case of a TSG $\stab$ and from the perspective of the Hamiltonian model, a string operator $p$ of charge $c$ transports a charge $c$ from one endpoint to another.
E.g., when applied to a ground state $p$ creates a pair of excitations with charge $c$ on its endpoints.

Our next goal is to study the commutation properties of strings, which for many geometries only depend on their charges. 
We start considering three strings $p,q,r$ common charge $c$ and a common endpoint, see \fref{commutation}(a). 
The quantity 
\begin{equation}\label{spin_inf}
    \spin c:=\comm {pq}{pr} = \comm pq \comm pr \comm qr
\end{equation}
only depends on the charge $c$, so that it is well defined.
To check this, consider an alternative set of strings $p',q', r'$ as in \fref{commutation}(a), together with the auxiliary strings $u,v,w$. 
Let $s= pqp' q' uv$ and $t= prp' r' uw$. 
Then $\comm s t = \comm {pq}{pr} \comm {p' q'}{p' r'}$, but $s,t\in \cen\gauge\pauli\propto \stab$ and thus $\comm s t=1$ because $\stab$ is abelian. 
We can interpret \eqref{spin_inf} as the topological spin of the charge $c$ \cite{levin:2003:fermions}. 
We say that a charge $c$ is bosonic if $\spin c=1$ and fermionic if $\spin c=-1$.

Now take two crossing strings $p,q$ with charges $c,c'$, as in \fref{commutation}(b). 
The mutual statistics of the charges $c$, $d$ is given by the quantity \cite{levin:2003:fermions}
\begin{equation}\label{mutual}
    \mutual c d:=\comm p q.
\end{equation}
It only depends on the charges $c,d$, so that it is well defined. 
The proof is analogous to the previous one, now with the geometry of \fref{commutation}(b).
If $\mutual c d=-1$ mutual statistics are said to be semionic, otherwise they are trivial.
If in \fref{commutation}(c) $p$ has charge $c$, $q$ has charge $d_1$ and $r$ has charge $d_2$, it follows that
\be
\mutual c{d_1d_2}=\mutual c {d_1}\mutual c {d_2},
\ee
since $\mutual c{d_1d_2}=\comm {p}{qr}=\comm {p}{q}\comm {p}{r}=\mutual c {d_1}\mutual c {d_2}$. 

The topological spin and mutual statistics are related. 
Consider the geometry of \fref{commutation}(d), where the strings $p,q,r$ have charge $c$ and the strings $p', q', r'$ have charge $d$. 
The figure illustrates that
\begin{equation}
    \spin {cd}=\spin c\spin {d}\mutual{c}{d},
\end{equation}
since $\spin {cd}=\comm {pp'q q'}{pp'r r'}=\comm {pq}{pr}\comm {p' q'}{p' r'}\comm r {q'}=\spin c\spin {d}\mutual{c}{d}$. 
In particular $\mutual c c=1.$

In summary, the charges in $\charge\gauge$ can be regarded as the topological charges of an anyon model, with fusion given by the charge group product.
But the topological structure does not stop there, because it is possible to define $\mutualx$ for a string with charge in $\charge \stab$ and another with charge in  $\charge \gauge$.
This works because $\cen\stab\pauli \subset\cen{\cen\gauge\pauli}\pauli$, which is all we need for the construction in \fref{commutation}(b) to make sense.
We cannot define $\mutualx$ for two strings with charges in $\charge\stab$, nor can we define $\spinx$ for charges in $\charge\stab$, unless of course we are dealing with a TSG rather than a general $TSSG$.
Therefore, we should regard the charges in $\charge\gauge$ as anyons, and the charges in $\charge\stab$ as fluxes that interact topogically with these anyons through an Aharonov-Bohm effect.

There is a natural morphism $\funct {\charge\iota} {\charge \gauge}{\charge\stab}$ derived from the restriction of morphisms from $\gauge$ to morphisms from $\stab$.
It preserves commutation in the sense that $\mutual c d =\mutual {\charge\iota (c)}{d}$ for $c, d\in\charge\gauge$.
This is an important ingredient in the construction of canonical charge generators given in \thmref{charges}:
\begin{align}
\charge\gauge=\sget {c_1,\dots,c_\alpha, d_1,\dots,d_\alpha,e_1,\dots,e_\beta},\qquad\quad\tilde c_i = \charge\iota (d_i), \nl
\charge\stab=\sget {\tilde c_1,\dots,\tilde c_\alpha, \tilde d_1,\dots,\tilde d_\alpha, \tilde e_1,\dots,\tilde e_\beta},\qquad\quad \tilde d_i = \charge\iota (c_i),
\end{align}
where $i=1,\dots,\alpha$, the $e_i$ generate the kernel of $\charge\iota$ and all gauge charge generators are bosons except possibly $e_1$, $c_1$ and $d_1$, always with $\spin{c_1}= \spin{d_1}$.
Among gauge charge generators the only nontrivial mutual statistcs are $\mutual {c_i}{d_i}=-1$.
The mutual statistics between gauge and stabilizer charges gives rise to the duality.
Indeed, relabeling $\charge\gauge=\sget{c_1,\dots,c_{2\alpha+\beta}}$, $\charge\stab=\sget{\tilde c_1,\dots,\tilde c_{2\alpha+\beta}}$,
\be
\mutual {c_i}{\tilde c_j}=1-2\delta_{ij}, \qquad i,j=1,\dots,2\alpha+\beta.
\ee

The anyon model attached to $\charge\gauge$ only has a few parameters, $\alpha$, $\beta$, $\chi=\spin{c_1}$ and $f=\spin{e_1}$.
We call this the characteristic of the TSSG.
The anyon model is chiral when $\chi=-1$, and when $f=-1$ there are fermions among the $e_i$.
There exist TSSGs covering all possible combinations of these parameters, as shown in \secref{canonical}.
Since the subsystem toric code only has a nontrivial element in $\charge\gauge$ it is an example of a TSSG with $\alpha=0$ and $\beta=1$.
Moreover, string operators are products of single qubit $Z$ Pauli operators, and thus $f=1$.

In the case of TSGs $\beta=0$, leaving the parameters $\alpha$ and $\chi$.
The toric code has $\alpha=1$ and $\chi=1$, and thus combining $N$ toric codes together we get a code with $\alpha=N$ and $\chi=1$.
Chiral codes are not likely to exist \cite{kitaev:2003:ftanyons}.

\subsection{Code structure and homology}

\begin{figure}
\psfrag{a}{\scriptsize$c_1$}
\psfrag{b}{\scriptsize$c_\alpha$}
\psfrag{c}{\scriptsize$d_1$}
\psfrag{d}{\scriptsize$d_\alpha$}
\psfrag{e}{\scriptsize$e_1$}
\psfrag{f}{\scriptsize$e_\beta$}
\psfrag{g}{\scriptsize$\tilde e_i$}
\begin{center}
\includegraphics[width=.45\columnwidth]{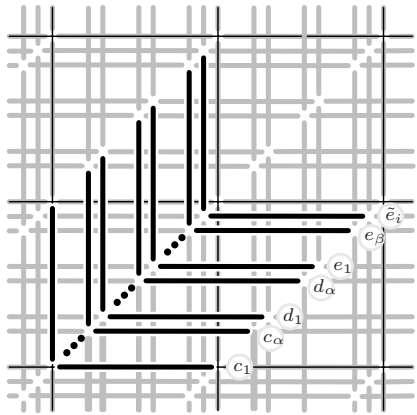}
\end{center}
\caption{
Framework of string segments. Each charge contributes a square lattice.
}
\label{fig:framework}
\end{figure}

With string operators and their geometrical commutation relations at hand, it is easy to uncover the general structure of TSSGs.
Given a TSSG, we can put it in a standard form ---its prototype is the toric code--- through the following process.
The first step is to find a set of canonical charges.
Then we construct a translationally invariant framework of string operators, each in the form of a straight segment.
For each charge there is a square lattice of segment operators, spatially disposed as indicated in \fref{framework}.
With suitable adjustments, the commutation relations between these segments operators are fixed by $\spinx$ and $\mutualx$.
Strings with charge $c_i$ ($d_i$) \emph{also} have charge $\tilde d_i$ ($\tilde c_i$), and we adjust segment operators with charge $e_i$ so that they are elements of $\gauge$
---in the technical part we impose more structure on those with charge $\tilde e_i$, but it does not affect the present discussion, only that of local equivalence.

\begin{figure}
\psfrag{x1}{$X_{2i-1}$}
\psfrag{x2}{$X_{2i}$}
\psfrag{z1}{$Z_{2i-1}$}
\psfrag{z2}{$Z_{2i}$}
\begin{center}
\includegraphics[width=.9\columnwidth]{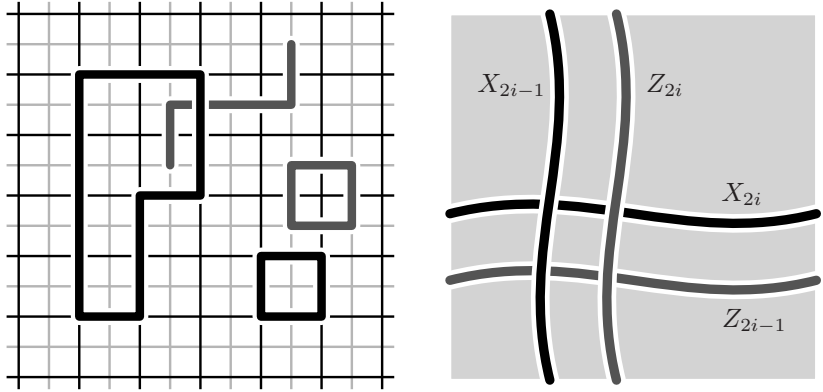}
\end{center}
\caption{
(Left) A square lattice and its dual.
String operators can be visualized as collections of segments on these two lattices.
We show for such strings: a direct plaquette, a dual plaquette, a direct boundary and a dual open string.
The last two cross once.
(Right) Logical operators corresponding to strings with charge $c_i$, in black, and strings with charge $d_i$, in grey, for any $i=1,\dots,\alpha$.
The topology is that of a torus, with opposite boundaries identified.
}
\label{fig:strings}
\label{fig:torus}
\end{figure}

Rather than the intricate framework of \fref{framework}, it is enough to visualize, for each pair of dual charges $c_i,d_i$ (or $e_i,\tilde e_i$) a single square lattice and its dual, see \fref{strings}.
Segment operators with charge $c_i$ or $e_i$ belong to the direct lattice, and those with charge $d_i$ or $\tilde e_i$ to the dual lattice.
We can form string operators by taking products of segment operators with the same charge.
Operators from different lattices commute, and the commutation properties of string operators are simply dictated by $\mutualx$ and $\spinx$.
E.g., a direct string with charge $c_i$ anticommutes with a dual string of charge $d_i$ iff the two strings cross an odd number of times.

Closed strings operators ---or boundary strings since the homology is trivial in the plane--- belong to $\stab$ or at least to $\gauge$ ---the former is the case of $\tilde e_i$ closed strings.
Any boundary operator is a product of `plaquette operators', each of them the product of the segment operators forming a plaquette.
Crucially, it is possible to find LTI sets of independent generators of $\gauge$ and $\stab$ that include these plaquette operators, which are the only ones subject to constraints ---now obvious.
In particular, those plaquette operators formed by $\tilde c_i, \tilde d_i$ or $\tilde e_i$ ($c_i$, $d_i$ or $e_i$) strings are generators of $\stab$ ($\gauge$).

The trivial stabilizer generators, i.e., those that are not plaquette operators, play no significant role, in the sense that it is always possible to find a local transformation that removes them as disentangled qubits. 
After this removal, we are left with a system where every Pauli operator can be decomposed as $sg$, where $s$ is a product of string operators and $g$ is a product of trivial gauge generators.

\subsection{Back to codes}

The homological perspective on string operators becomes most relevant when we go back to a code in a finite lattice, as we sketch here.
The simplest possible way to do this is by considering a finite lattice with periodic boundary conditions, so that the topology is that of a torus.
The decomposition in terms of string operators and trivial gauge generators still holds ---the lattice must be large enough, though, in terms of the support of the segment operators and trivial gauge generators.
But now closed string operators need not be boundaries, and this is an all-important feature.

Regarding gauge and stabilizer generators, each `global constraint' becomes now a constraint of the form $\prod_i p_i \propto 1$, with $p_i$ running over plaquette operators of a given lattice, direct or dual, so that we can remove one such plaquette from the corresponding generating sets.
But this still does not give rise in general to a proper stabilizer and gauge group pair.
Indeed, when $\beta\neq 0$ we have to either add nontrivial cycle strings with charge $e_i$ to $\stab$ or nontrivial cycle string with charge $\tilde e_i$ to $\gauge$.
This does not affect the number of encoded qubits ---but see below!.
In fact, apart from this detail everything works as in a toric code.
Each pair $(c_i,d_i)$ contributes two logical qubits, so that $k=2\alpha$, and logical operators correspond to nontrivial cycles, see \fref{torus}.

From the above discussion it seems that the existence of the $\beta$ charges $e_j$ serves no purpose from a coding perspective. 
For quantum codes this is indeed true in the sense that they contribute no encoded qubits. However, interestingly they do provide encoded \emph{bits}! 
The eigenvalues of nontrivial closed strings with charge $\tilde e_i$ label the classical states of the memory ---and should be regarded as elements of $\gauge$---, and nontrivial closed string operators with charge $e_i$ act as bit flip operators. 
Detecting errors thus amount to look for the endpoint of open string operators with charge $e_i$. I.e., the stabilizer ---which must include only strings of trivial homology---, still provides the error syndrome through the eigenvalues of its generators. 
The subsystem toric code exemplifies all this.

\subsection{Equivalence}

We say that two TSSGs are locally equivalent if they can be related by a combination of LTI Clifford transformations, lattice coarse graining and addition/removal of disentangled qubits in a translationally invariant way.
Such local transformations preserve the topological charge structure, and the converse also holds in some cases, see \secref{equivalence}:
\mybox{Non-chiral TSSGs are locally equivalent iff they have the same characteristic.}
This result, true also for chiral TSGs, is a consequence of the structure discussed above, together with the existence of certain `elementary' TSSGs and the ability to `complete'
partially defined LTI Clifford transformations, see \lemref{extension}.

\subsection{Centralizers and 2D topology}\label{sec:centralizer}

\begin{figure}
\psfrag{(a)}{\,\,\,(a)}
\psfrag{(b)}{\,\,\,(b)}
\psfrag{(c)}{\,\,\,(c)}
\psfrag{(d)}{\,\,\,(d)}
\begin{center}
\includegraphics[width=.9\columnwidth]{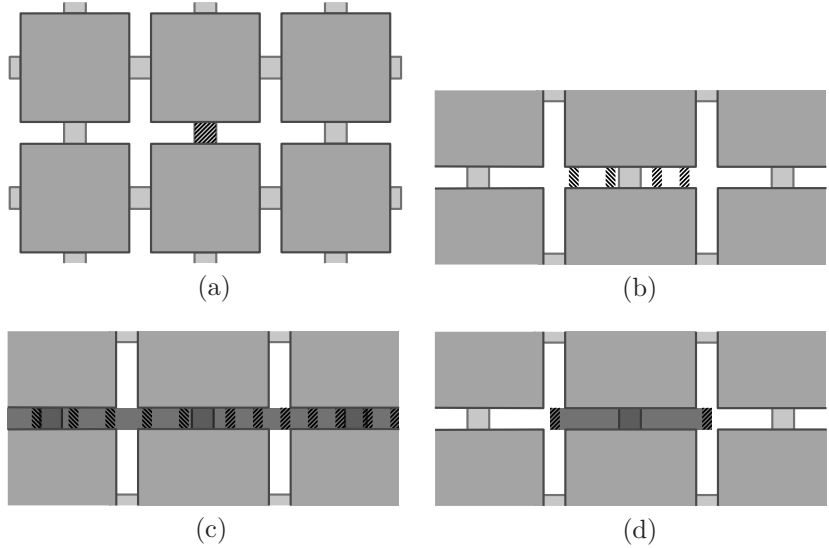}
\end{center}
\caption{
The regions involved in the construction of an LTI set of independent generators.
}
\label{fig:regions}
\end{figure}

To finish this summary of results, let us explain the role played by the topology of the plane and the fact that the LPG is a centralizer of a LPG for the existence of an LTI set of independent generators.
Let the support of a constraint be the union of the support of its elements.
Consider a LPG $\PA$ of the form $\PA=\cen\PB\pauli$.
In order to recover a LTI set of independent generators from an arbitrary LTI set of generators of $\PA$, all we need to do is remove in a translationally invariant way a sufficient amout of generators.
If a generator is part of a constraint, it can be removed.
The removal of constrained generators with support in the regions $A$, marked with dark grey in \fref{regions}(a), can be done without much difficulty.
It is enough to choose a large enough separation between these regions, so that generators in different regions are subject to independent constraints, which allows a translationally invariant removal.
After the removal it is unclear whether the generators with support in different regions $B$, marked with light grey in \fref{regions}(a), are also subject to independent constraints.
But given a constraint $R$ with support, say, in the region $B_0$ marked with stripes in \fref{regions}(a), we can find ---as shown below--- another constraint $R_0$ containing the same generators as $R$ in $B_0$ but with no support in the rest of regions $B$.
This shows that generators with support in different regions $B$ are indeed subject to independent constraints.
Then we are left with the regions $C$, marked with white in \fref{regions}(a), but these are disconnected from each other and we can deal with them as follows.
The constraints that are left can only include generators with support in regions $C$ now, and since they are distant enough and the generators are local, the restriction of any constraint to a particular region $C$ is still a constraint.
Thus, generators with support in different regions $C$ are also subject to independent constraints.

Thus, the crucial point is the existence, for any given constraint $R$ with support in regions $B$ and $C$, of a constraint $R_0$ with the same support as $R$ in $B_0$ but no support in other regions $B$.
To guarantee it, we take the size of the regions $A$ as large as needed, so that near $B_0$ any such $R$ has support only in a long and narrow band.
Then, up to translations, the elements in $R$ must repeat along the band at certain stretches of a length larger than the size of the generators.
At each side of the band we pick two such repeated stretches, marked with stripes in \fref{regions}(b), and form a new constraint $R_{\infty}$ that extendes to infinity by repeating the set of constraints in between the repeated stretches up to translations.
This is depicted in \fref{regions}(c), where the places where the repeated stretches overlap are marked with stripes and the support of $R_{\infty}$ is darkened.
We can then consider a subset of $R_{\infty}$, containing only the elements on the darkened region in \fref{regions}(d), which produces a set of generators $S$ with product $p_1p_2\in\pauli$, where $p_i$ have each support in one of the two regions marked with stripes.
The all-important property of these two regions is that they are sufficiently far away from region $A$.
Since $p_1p_2\in\cen\PB\pauli$ and $\PB$ has local generators, it follows that $p_1,p_2\in\cen\PB\pauli$ since the supports of $p_1$ and $p_2$ are sufficiently far appart.
There exist sets of generators $S_i$ of $\PA$, $i=1,2$, with supports in a small neighborhood of the support of $p_i$ and such that their product is proportional to $p_i$.
Then to obtain $R_0$ it is enough to take the symmetric difference of $S$, $S_1$ and $S_2$.

\section{Basic definitions}\label{sec:basics}

This section intends to put together most of the constructions needed later, with the aim of both presenting them and serving as a reference for later sections.
It includes a list of basic examples of topological codes.

\subsection{Notation}

\begin{sep}{Sets} 
For integers we use the notation $\N_0:=\sset 0 \cup \N^\ast$. $|S|$ denotes the cardinality of a set $S$, $\powerset S$ its power set and $\powersetfin S:=\set{s\in \powerset S}{|s|<\aleph_0}.$
Given a mapping $\funct f A B$ and any subsets $a\subset A$, $b\subset B$, we denote by $f[a]\subset B$ the image of $a$ and by $f\inv[b]\subset A$ the preimage of $b$.
Given any set $S$, we regard $\powerset S$ as the abelian group with addition given by the symmetric difference $S_1\add S_2:=S_1\cup S_2-S_1\cap S_2.$
\end{sep}

\begin{sep}{2D lattice} 
Since we are interested in topological aspects, we consider without loss of generality square lattices.
Sites are labeled as usual with integer coordinates $\site j k\in\Z^2$, with $\Sigma=\Z^2$ the set of all sites, and $\Site:=\site 0 0$.
A block is a set of sites of the form 
\be
\bl m n {m'} {n'} := \set{\site j k} {m\leq j< m'; n\leq k< n'}.
\ee
Given a set of sites $\gamma\in\powersetfin\Sigma$, $\range\gamma$ denotes the smallest $k\in \N_0$ such that $\gamma\subset\bl m n {m+k}{n+k}$ for some $m,n\in\Z$.
For $\gamma\subset\Sigma$, $l\in\N_0$, define a `thickened' set 
\be
\thk^l (\gamma) :=\bigcup_{\site i j\in\gamma} \bl {i-l}{j-l} {i+l+1}{j+l+1}.
\ee
\end{sep}

\begin{sep}{Connectedness}
Two sites $\site i j$ and $\site m n$ are adjacent if $|m-i|\leq 1$ and $|n-j|\leq 1$.
An ordered list of sites $(\sigma_k)_{k=1}^n=(\sigma_0,\cdots,\sigma_n)$ where $\sigma_k$ is adjacent to $\sigma_{k+1}$ for all $k=1,\dots,n$ is a path from $\sigma_0$ to $\sigma_n$.
We define the inverse $\gamma\inv$ of a path $\gamma$ in the usual way, an also the composition of paths $\gamma'\circ\gamma$, where the last site of $\gamma$ is the first site of $\gamma'$.
A set of sites $\gamma\subset\Sigma$ is connected if for any $\sigma,\sigma'\in\gamma$ there exists a path $(\sigma_k)_{k=1}^n$ with $\sigma_0=\sigma$, $\sigma_n=\sigma'$ and $\sigma_i\in\gamma$, $i=1,\dots,n$.
Given an ordered set of sites $(\sigma_k)_{k=1}^n=((a_k,b_k))_{k=1}^n$, let $\sigma_k':=(a_{k+1},b_k),$ $k=1,\dots,n-1$ and define 
$
\pth {\sigma_1,\dots,\sigma_n}
$
to be the path starting at $\sigma_1$, moving in a straight line to $\sigma_1'$, then in a straight line to $\sigma_2$, then $\sigma_2'$,\dots, and ending at $\sigma_n$.
Sometimes we identify a path $(\sigma_i)$ with the set $\sset{\sigma_i}$.
\end{sep}

\begin{sep}{Homology}
Let $\graph$ be an infinite square lattice.
$\faces$, $\dfaces$, $\edges$ and $\dedges$ denote, respectively, the set of faces, dual faces, edges and dual edges.
We represent $\Z_2$ chains of these objects as elements of $\powerset\faces$, $\powerset\dfaces$, $\powerset\edges$ and $\powerset\dedges$.
$\partial$ denotes the usual boundary operation, which we express in particular as
\begin{align}\label{homology}
\morphismd {\partial}{\partial}{\powerset\faces}{\powerset \edges}{\powerset\dfaces},\qquad
\morphismd {\partial}{\partial}{\powerset\dfaces}{\powerset \dedges}{\powerset\faces}.
\end{align}
Given an edge $e$, it dual is denoted $e^\ast$.
We name the elements $\hedge,\vedge\in\edges$, $\face\in\faces, \face^\ast\in\dfaces$ according to \fref{lattice}.
For faces, the symbol $^\ast$ is just a label.
\end{sep}

\begin{figure}
\psfrag{f}{$\face$}
\psfrag{fa}{$\face^\ast$}
\psfrag{e1}{$\hedge$}
\psfrag{e2}{$\vedge$}
\psfrag{e1a}{$\hedge^\ast$}
\psfrag{e2a}{$\vedge^\ast$}
\begin{center}
\includegraphics[width=.45\columnwidth]{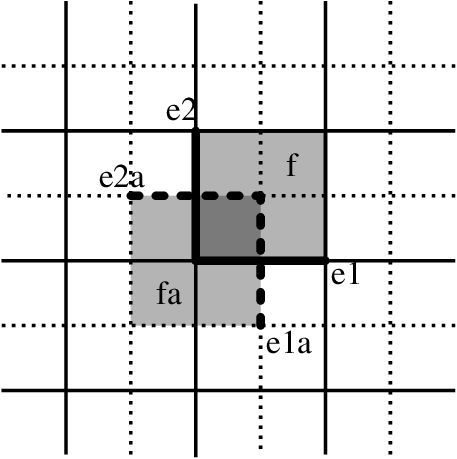}
\end{center}
\caption{
Naming some elements of $\graph$.
}
\label{fig:lattice}
\end{figure}

\begin{sep}{Qubits and operators}
Sites hold a common and finite number of qubits $n$.
Let $X_{j,k}^\alpha$, $Z_{j,k}^\alpha$ denote single qubit $X$ and $Z$ operators acting on the qubit with label $\alpha=1,\dots,n$ on the site $\site j k$. 
They generate for each site $\site j k$ an algebra of operators $A_{\site j k}$.
Consider the infinite tensor product $A:=\otimes_{\site j k} A_{\site j k}$, with elements $\otimes_{\site j k} a_{\site j k}$ with all but a finite number of the product terms identical $a_{\site j k}=\id$.
The Pauli group $\pauli\subset A$ is $\pauli:=\phase\sget{X_{ij}^\alpha,Z_{ij}^\alpha}_{i,j,\alpha}.$
Given a Pauli operator $p\in\pauli$, $\supp p$ denotes the set of sites containing qubits on the support of $p$, i.e., on which $p$ acts nontrivially.
Its size is measured by $\range p :=  \range{\supp p}$.
For $\gamma\subset\Sigma$ we let $p|_\gamma$ be the restriction of $p$ to qubits in $\gamma$, defined in the obvious way and only up to a global phase. 
Given $p,q\in\pauli$ their commutator is $\comm p q:= pqp\inv q\inv=\pm 1.$
For $\mathcal A,\mathcal B\subset \pauli$, the centralizer of $\mathcal A$ in $\mathcal B$ is
\be
\cen {\mathcal A}{\mathcal B}:=\set {b\in\mathcal B}{\forall a\in \mathcal A \quad \comm a b=1 }.
\ee
We write $\mathcal A\propto \mathcal B$ meaning $\phase\mathcal A=\phase\mathcal B$.
\end{sep}

\begin{sep}{Pauli group subsets}
For $S\subset\pauli$, $\gamma\subset\Sigma$ and $q\in\pauli$ we define
\begin{align}
\supp S:=\bigcup \suppi S,\qquad&
S|_{\gamma}:=\set{s\in S}{\supp s\subset\gamma},\nl
S\overlap \gamma:=\set{s\in S}{\supp s\cap\gamma\neq\emptyset},\qquad&
\comm S q := \prod_{p\in S_q } \comm p q,
\end{align}
where $S_q=\set{p\in S}{\supp p\cap\supp q\neq\emptyset}$ and only for $S_q$ finite the definition holds.
We set $\range S := \range{\supp S}.$
We say that a set of generators $\PA_g$ of a group $\PA\subset \pauli$ is independent if the product of any finite subset of nontrivial generators is nontrivial 
---trivial elements of $\pauli$ are those in $\phase$.
We define a `product' group morphism
\functmap {\prox}{\powersetfin \pauli} {\pauli/\phase}{S}{ \phase\prod_{p\in S}p.}
\end{sep} 

\begin{sep}{Morphisms} 
Let $\PA\subset\pauli$ be a group. The group morphisms $\funct \phi \PA \sset{1,-1}$ such that $\phi(i)=1$ form an abelian group $\morph {\PA}$ with
$
(\phi\phi')(a):=\phi(a)\phi'(a).
$
Given a subset $\PA_g$ of $\PA$, we define a group morphism
\functmap  {\negrx {\PA_g}} {\morph \PA}{\powerset{\PA_g-\phase}}{\phi}{\PA_g\cap\phi^{-1}[-1],}
which has an inverse $\morx{\PA_g}$ if $\PA_g$ is an independent set of generators of $\PA$.
The support of $\phi\in\morph\PA$ with respect to $\PA_g$ is
\be
\suppr {\PA_g} \phi := \supp {\negr {\PA_g} \phi}.
\ee
We define the subgroup of morphisms with finite support
\be
\morphfin {\PA_g}:=\set{\phi\in\morph {\sget{\PA_g}}}{|\suppr {\PA_g}\phi|<\aleph_0}.
\ee
Given $\phi\in\morphfin {\PA_g}$, we will denote with the same name the group morphism 
\functmap \phi {\powerset {\PA_g}}{\sset{1,-1}}{S}{(-1)^{| S \cap \negr{\PA_g}\phi|}}
The commutation structure gives rise to morphisms that need a name,
\functmapdouble {\syndrx{\PA}}{\pauli} {\morph {\PA}}{p}{ \comm p\cdot} {\powerset\pauli} {\morph {\PA}}{S}{ \comm S\cdot}
For $p\in \pauli$ or $p\in\powerset{\pauli}$ we set $\suppr{\PA_g}p:=\suppr {\PA_g}{\comm p\cdot}.$
\end{sep}

\begin{sep}{Translation operators} 
We define translation operators for sites setting
\be
\trans m n \site j k:=\site {j+m}{k+n}.
\ee
The action on edges or faces of $\graph$ is analogous.
For Pauli operators we define $\funct{\trans m n}{\pauli}{\pauli}$ as group morphisms with $\trans m n (i\id):=i\id$ and
\be
\trans m n (\sigma_{j,k}^\alpha):=\sigma_{j+m,k+n}^\alpha,\qquad \sigma=X,Z.
\ee
For $\phi\in\morph \PA$ we set
\be
(\trans mn (\phi)) (\cdot):=\phi( \trans {-m}{-n} (\cdot)).
\ee
To generalize translation to sets of translatable objects, we recursively define whenever it makes sense
\be
\trans m n (A) :=\set {\trans m n (a)}{a\in A}.
\ee
Finally, we set for $d\in\N^\ast$
\be
\strans d (\cdot):=\bigcup_{m,n\in\Z} \trans {md} {nd} (\cdot).
\ee
\end{sep}

\begin{sep}{Coarse graining and composition} 
A qubit lattice can be coarse grained by identifying each block 
$\bl {ml} {nl} {(m+1)l} {(n+1)l}$
as the site $\site m n$ of the new lattice, for some $l\in\N^\ast$.
We denote the effect of coarse graining on a given object $x$ by $\coarse l x$, be it an operator, a set of sites, a group or any other structure.

Two disjoint qubit lattices can be put together to form a single lattice.
If $\PA$ is a Pauli subgroup on the first lattice and $\PB$ is a Pauli subgroup on the second lattice we define a new group
$
\PA\otimes\PB:=\set{a\otimes b}{a\in\PA,b\in \PB}.
$
\end{sep}

\subsection{Lattice Pauli morphism}

Consider Pauli groups $\pauli_i$ over square lattices.

\begin{defn}
A lattice Pauli (iso)morphism  is a group (iso)morphism  $\funct F{\pauli_{1}}{\pauli_{2}}$
such that $F(i)=i$ and, for any $m,n\in\Z$ and some $d\in\N^\ast$,
\begin{equation}
\trans {md} {nd} \circ F = F\circ \trans {md} {nd}.
\end{equation}
We say that $F$ has period $d$ and define
\be
\range F := \max_{\sigma\in\bl 0 0 dd} \range {\sset\sigma\cup\supp{F[\pauli_1|_{\sset\sigma}]}}.
\ee
\end{defn}
Such a morphism is always injective. $\range F$ is finite and
\be
\supp{F(p)}\subset\thk^{\range F-1}(\supp p)
\ee
for every $p\in\pauli_1$.
By coarse graining we can always have $d=1$ and $\range F = 2$.
If $n_i$ is the number of qubits per site in $\pauli_i$, $n_1\leq n_2$ and the equality holds iff $F$ is an isomorphism.
To see way, take $\gamma=\bl 0 0 L L$ and let $g_1=2n_1L^2$ be the number of generators of $\pauli_1':=\pauli_1|_{\gamma}$ and $g_2=2n_2(L+2\range F-2)^2$ be the number of generators of $\pauli_2':=\pauli_2|_{\thk^{\range F-1}(\gamma)}$.
Then $F[\pauli_1']\subset\pauli_2'$ gives $g_1\leq g_2$ and thus in the limit of large $L$ we get the desired result.
The following `extension lemma' will be very useful later and is true for any spatial dimensionality.

\begin{lem}\label{lem:extension}
For every lattice Pauli morphism $\funct F{\pauli_{1}}{\pauli_{3}}$ there exist $\pauli_2$ and a lattice Pauli isomorphism $\funct {F_{\mathrm ext}}{\pauli_{1}\otimes\pauli_{2}}{\pauli_{3}}$ with the same period and
\be
F_{\mathrm ext} (p\otimes\id) = F(p).
\ee
\end{lem}

Since $F_{\mathrm ext}$ preserve centralizers, there exists $\funct G {\pauli_2}{\pauli_3}$ with
\be\label{extension_FG}
F_{\mathrm {ext}}(p\otimes q)= F(p)G(q),\quad G[\pauli_2]=\cen{F[\pauli_1]}{\pauli_3}.
\ee

\begin{proof}
In this proof we need to deal with `regions' that are collections of qubits, rather than of complete sites.
So let $\Sigma(\alpha)$ denote all qubits with label $\alpha$ and $\pauli_i^\alpha:=\pauli_i|_{\Sigma(\alpha)}$.
We can assume w.l.o.g. $d=1$.
To avoid trivial cases we assume $1\leq n_1<n_3$, and set $n_2=n_3-n_1$.
If we can construct a lattice Pauli isomorphism $\funct {F'}{\pauli_{3}}{\pauli_{1}\otimes\pauli_2}$ with period 1 and such that $F'(F(\cdot)) = F''(\cdot)\otimes\id$ for some isomorphism $\funct {F''}{\pauli_{1}}{\pauli_{1}}$, the result follows taking $F_{\mathrm ext}= {F'}\inv\circ(F''\otimes {\bf I}_{\pauli_2})$.
We construct $F'$ as a composite $F'=F_b\circ F_a$. 
In turn, $\funct {F_a}{\pauli_{3}}{\pauli_{3}}$ is composed of two kinds of maps. 
The first kind is that of single qubit Clifford gates, preserving translational invariance.
The second corresponds to CNot gates applied to pairs of qubits with different label $\alpha$ ---but possibly in different sites--- again preserving translational invariance.
Composing such `moves', it is a routine computation ---we are essentially back to a finite Pauli group case, with the restriction of not being able to perform CNot gates between qubits with the same label--- to construct an isomorphism $F_1$ such that $F_1'[\pauli_1^1|_{\Site}]\subset\pauli_3^1$ for $F_1':=F_1\circ F$.
Then we have $F_1'[\pauli_1^1]=\pauli_3^1$ ---translational invariance only implies inclusion, but the restriction of $F_1'$ to $\pauli_1^1$ and $\pauli_3^1$ is actually an isomorphism, since $\pauli_1^1$, $\pauli_3^1$ can be regarded as lattices with $1$ qubit per site. 
Now let $u\in F_1'[\pauli_1^1]$, $v\in F_1'[\pauli_1^\alpha]$, $1<\alpha\leq n_1$. Then $1=\comm u v = \comm u{v|_{\Sigma(1)}}$ and thus $v|_{\Sigma(1)}\subset\cen {\pauli_3^1} {\pauli_3^1}\propto\phase$. 
The idea is then to repeat the procedure and successively construct for $i=2,\dots,n_1$, isomorphisms $F_i$ from $F_{i-1}$ without using any elementary moves involving qubits with labels $\alpha<i$, so that $(F_i\circ F)[\pauli_1^\beta]=\pauli_3^\beta$ for $\beta\leq i$.
It suffices then to set $F_a=F_{n_1}$ and let $\funct {F_b}{\pauli_{3}}{\pauli_{1}\otimes\pauli_2}$ be the isomorphism such that $F_b(\sigma_{i,j}^\alpha\sigma_{i,j}^{n_1+\beta})=\sigma_{i,j}^\alpha\otimes\sigma_{i,j}^{\beta}$, $\alpha=1,\dots,n_1$, $\beta=1,\dots,n_2$.
\end{proof}

\subsection{Lattice Pauli groups}

\begin{defn}
Let $S$ be a set such that $\suppi S$ and $\trans m n [S]$ are defined.
$S$ is local and translationally invariant (LTI) if there exist $k,l\in\N^\ast$ such that for any $s\in S$, $m,n\in\Z$,
\be\label{locality}
 \range s \leq k,\quad\qquad \trans {ml} {nl} (s)\in S.
\ee
In that case we say that $S$ is $k$-bounded and has period $l$.
\end{defn}

\begin{defn}
A subgroup $\PA\subset\pauli$ is a lattice Pauli group (LPG) if it admits a LTI set of generators $\PA_g\subset\PA$.
\end{defn}

When we say that $\PA_g$ generates a LPG $\PA$ it is understood that $\PA_g$ is a LTI set.
The following result is an immediate consequence of the definition of LPGs.

\begin{prop}\label{prop:coarse_grain1}
Consider a set of LPGs $\bar \PA^1,\dots,\bar\PA^n$ on a given lattice. There exists $L\in\N^\ast$ such that  the coarse grained LPGs 
\be
\PA^k=\coarse L {\bar \PA^k}\subset\pauli, \qquad k=1,\dots,n,
\ee
admit LTI sets of independent generators $\PA^k_g$ such that
\begin{enumerate}
  \item $\PA^i_g$ has period 1, and
  \item $\range a\leq 2$ for any $a\in\PA^i_g$. 
\end{enumerate}
\end{prop}

We will assume that all the LPGs that we discuss satisfy these properties, unless otherwise stated.

\begin{defn}
Given two LPGs $\PA,\PB$ defined on disjoint qubit lattices, their composition is the LPG $\PA\otimes\PB$ defined on the union of the two lattices. 
\end{defn}
$\PA\otimes\PB$ is indeed an LPG because $\PA_g\cup\PB_g$ is a valid LTI generating set.

\begin{defn}
A LPG (iso)morphism from $\PA_1\subset\pauli_1$ to $\PA_2\subset\pauli_2$,
denoted $\morphism F {\PA_1}{\PA_2},$ is a lattice pauli (iso)morphism $\funct F{\pauli_{1}}{\pauli_{2}}$ with $F[\PA_1]=\PA_2.$
\end{defn}
Two LPGs $\PA_1$ and $\PA_2$ are isomorphic if there exists such an isomorphism $F$.
Given a LPG morphism $\morphism F {\PA_1}{\PA_2}$ we can coarse grain both lattices, that of $\PA$ and that of $\PB$, and obtain a morphism
\be
\morphism {\coarse l F} {\coarse l {\PA}}{\coarse l {\PB}},
\ee
where $l\in\N^\ast$.
Given two LPG morphisms $\morphism {F} {\PA_1} {\PA_2}, \morphism {G} {\PB_1} {\PB_2},$ if the corresponding lattices are disjoint we form in the obvious way the morphism
\be
\morphism {F\otimes G} {\PA_1\otimes\PB_1}{\PA_2\otimes\PB_2}.
\ee

\subsection{Topological stabilizer groups}

\begin{defn}
A lattice stabilizer group is a LPG $\stab\subset\pauli$ such that $-1\not \in\stab.$
\end{defn}

\begin{defn}
A topological stabilizer group (TSG) is a lattice stabilizer group $\stab\subset\pauli$ such that $\cen {\stab}{\pauli} \propto \stab.$
\end{defn}

\begin{defn}
A topological stabilizer subsystem group (TSSG) is a lattice stabilizer group $\stab\subset\pauli$ such that $\cen {\cen \stab\pauli}{\pauli} \propto \stab.$
Its gauge group is $\gauge:=\cen\stab\pauli.$
\end{defn}

TSGs are TSSGs.
As we will show in \corref{centralizer}, $\gauge$ is a LPG.
The product $\stab_1\otimes\stab_2$ of two TSSGs is also a TSSG.
A LPG isomorphism $F$ from $\stab_1$ to $\stab_2$ is also a LPG isomorphism from $\gauge_1$ to $\gauge_2$.

\subsection{Examples}

\begin{figure}
\psfrag{h}{$h$}
\psfrag{v}{$v$}
\psfrag{1}{1}
\psfrag{2}{2}
\psfrag{3}{3}
\psfrag{4}{4}
\psfrag{5}{5}
\psfrag{6}{6}
\psfrag{(a)}{(a)}
\psfrag{(b)}{(b)}
\psfrag{(c)}{(c)}
\begin{center}
\includegraphics[width=.9\columnwidth]{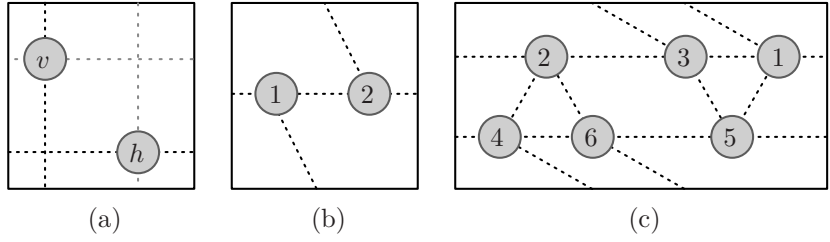}
\end{center}
\caption{
Unit cells for the TPGs in the examples. 
(a) In the toric code, qubits can be identified with horizontal and vertical links in a square lattice, or with vertical and horizontal links in its dual.
Stabilizers are related to faces and dual faces.
(b) In the honeycomb subsystem code qubits can be identified with the vertices of a honeycomb lattice.
Gauge generators are related to links.
(c) In the topological subsystem color code qubits can be identified with the vertices of a lattice derived from the honeycomb.
The numbering is obtained by moving counterclockwise along any hexagonal plaquette.
Gauge generators are related to links.
}
\label{fig:examples}
\end{figure}

Some examples of TSSGs follow. 
Among them, two subsystem variants of the toric code that we will find extremely useful.
All the generator sets below are independent.
Lattices are displayed in \fref{examples}.

\begin{sep}{Empty code} 
The simplest TSG, with no qubits.
\end{sep}

\begin{sep}{Trivial code} 
This TSG is constructed on a lattice with one qubit per site, with $\paulitrivial$ denoting the pauli group.
The stabilizer $\stabtrivial$ has generators
\be
\stabtrivialg:= \strans 1 (\sset{Z_{0,0}}).
\ee 
\end{sep}

\begin{sep}{Subsystem trivial code} 
This TSSG is constructed on the same lattice as the trivial code.
The stabilizer is trivial, $\stabstrivial=\sset\id$ and $\gaugestrivial=\paulitrivial$.
\end{sep}

\begin{sep}{Toric code} 
This TSG  \cite{kitaev:2003:ftanyons} is constructed on a lattice with two qubits per site, labeled as $h$ and $v$.
We denote by $\pauliTC$ the corresponding pauli group.
The stabilizer $\stabTC$ has generators
\begin{align}\label{TC_stab}
\stabTCg:= \strans 1 (\sset{S^X,S^Z}),&\qquad
S^Z:=Z_{0,0}^h Z_{0,1}^h Z_{0,0}^v Z_{1,0}^v,\nl
S^X:=X_{0,0}^h X_{-1,0}^h X_{0,0}^v X_{0,-1}^v.&
\end{align}
\end{sep}

\begin{sep}{Subsystem toric code} 
This TSSG is constructed on the same lattice as the toric code.
The stabilizer $\stabSTC$ and the gauge group $\gaugeSTC$ have generators
\be
\stabSTCg:= \strans 1 (\sset{S^Z}),\qquad
\gaugeSTCg:= \phase\strans 1 (\sset{S^X,Z_{0,0}^h,Z_{0,0}^v}).
\ee 
\end{sep}

\begin{sep}{Fermionic subsystem toric code} 
This TSSG is constructed on the same lattice as the toric code.
$\stabFTC$ and $\gaugeFTC$ have generators
\begin{align}
\stabFTCg:= \strans 1 (\sset{S^Z_f}),&\qquad
\gaugeFTCg:= \phase\strans 1 (\sset{S^X,Z_f^h,Z_f^v}),\nl
Z_f^h:= Z_{0,0}^h X_{-1,0}^h X_{0,-1}^v X_{1,-1}^v, &\qquad
S^Z_f:= Z_{0,0}^h Z_{0,1}^h Y_{0,0}^v Y_{1,0}^vX_{1,0}^h X_{-1,1}^h,\nl
Z_f^v:= Z_{0,0}^v X_{-1,0}^h X_{0,-1}^v X_{0,0}^h.&
\end{align}
\end{sep}

\begin{sep}{Honeycomb subsystem code} 
This TSSG  \cite{suchara:2011:constructions} is constructed on a lattice with two qubits per site.
$\stabhon$ and $\gaugehon$ have generators
\begin{align}
\stabhong:= \strans 1 (\sset{S^{\hexagon}}),&\qquad
\gaugehong:= \phase\strans 1 (\sset{G^X_{0,0},G^Y_{0,0},G^Z_{0,0}}),\nl
G^X_{i,j}:=X_{i,j}^1X_{i,j}^2,&\qquad
S^{\hexagon}:=X_{0,0}^2Z_{1,0}^1Y_{1,0}^2X_{1,1}^1Y_{0,1}^1Z_{0,1}^2,\nl
G^Y_{i,j}:=Y_{i+1,j}^1Y_{i,j}^2,&\qquad
G^Z_{i,j}:=Z_{i,j+1}^1Z_{i,j}^2.
\end{align}
\end{sep}

\begin{sep}{Topological subsystem color code} 
This TSSG  \cite{bombin:2010:subsystem} is constructed on a lattice with six qubits per site, labeled 1-6.
$\stabTSCC$ and $\gaugeTSCC$ have generators
\begin{align*}
&\stabTSCCg:= \strans 1 (\sset{S^1_{0,0},S^2_{0,0}}),\qquad
\gaugeTSCCg:= \phase\strans 1 (\sset{G^1_{0,0},G^2_{0,0},\dots,G^{10}_{0,0}}),\nl
&S^2_{i,j}:=X_{i,j}^1Y_{i+1,j}^2X_{i+1,j}^3Y_{i+1,j+1}^4X_{i,j+1}^5Y_{i,j+1}^6 Y_{i,j+1}^1Y_{i,j+1}^2Y_{i,j}^3 \cdot\nl
&\quad\cdot Y_{i+1,j}^4Y_{i+1,j}^5Y_{i+1,j+1}^6X_{i+1,j}^1X_{i+1,j+1}^2X_{i,j+1}^3X_{i,j+1}^4X_{i,j}^5X_{i+1,j}^6,\nl
&S^1_{i,j}:=Z_{i,j}^1Z_{i+1,j}^2Z_{i+1,j}^3Z_{i+1,j+1}^4Z_{i,j+1}^5Z_{i,j+1}^6,\qquad G^{10}_{i,j}:=Z_{i,j}^2Z_{i,j}^6,
\end{align*}
\begin{align}
G^1_{i,j}&:=X_{i,j}^1Y_{i+1,j}^2,\quad & 
G^2_{i,j}&:=X_{i+1,j}^2Y_{i+1,j}^3,\quad&
G^3_{i,j}&:=X_{i+1,j}^3Y_{i+1,j+1}^4,\nl 
G^4_{i,j}&:=X_{i+1,j+1}^4Y_{i,j+1}^5,\quad&
G^5_{i,j}&:=X_{i,j+1}^5Y_{i,j+1}^6,\quad & 
G^6_{i,j}&:=X_{i,j+1}^6Y_{i,j}^1,\nl
G^7_{i,j}&:=Z_{i,j}^1Z_{i,j}^5,\quad & 
G^8_{i,j}&:=Z_{i,j}^3Z_{i,j}^5,\quad &
G^9_{i,j}&:=Z_{i,j}^2Z_{i,j}^4.
\end{align}
\end{sep}

\subsection{Local equivalence}

Two TSSGs are locally equivalent when they are the same up to local transformations, i.e., coarse graining, LPG isomorphisms and composition with a trivial TSG $\stabtrivial$ or a trivial TSSG $\stabstrivial$.

\begin{defn}
The TSSGs $\stab_1$, $\stab_2$ are locally equivalent if there exists a LPG isomorphism $\funct F {\pauli_1\otimes\paulitrivial^{k_1+m_1}}{\pauli_2\otimes\paulitrivial^{k_2+m_2}}$
\be\label{equivalence_s}
\morphism F {\coarse {l_1} {\stab_1}\otimes \stabtrivial^{\otimes k_1}\otimes 1}{\coarse {l_2} {\stab_2}\otimes\stabtrivial^{\otimes k_2}\otimes 1},
\ee
where $k_1,k_2,m_1,m_2\in\N_0$ and $l_1,l_2\in\N^\ast$.
\end{defn}
For convenience, call this a $(k_1,k_2,m_1,m_2,l_1,l_2)$-equivalence.
\begin{defn}
The TSGs $\stab_1$ and $\stab_2$ are locally equivalent if they are $(k_1,k_2,0,0,l_1,l_2)$-equivalent as TSSGs.
\end{defn}
Two TSGs that are equivalent as TSSGs are also equivalent as TSGs.
Indeed, if $F$ gives the TSSG equivalence it maps gauge elements not in the stabilizer to gauge elements not in the stabilizer.
I.e., denoting by $\paulitrivial^{\otimes m_i}$ the copies of $\paulitrivial$ holding $\stabstrivial^{\otimes m_i}$, we have $F[\paulitrivial^{\otimes m_1}]=\paulitrivial^{\otimes m_2}$.
But centralizers are mapped to centralizers, and the centralizer of $\paulitrivial^{\otimes m_i}$ is $\pauli_i\otimes\paulitrivial^{\otimes k_i}$, so that a restriction of $F$ will give the equivalence.

We still have to check transitivity.
But given a $(k_1,k_2,m_1,m_2,l_1,l_2)$-equivalence $F$ of $\stab_1$ and $\stab_2$, and a $(k_2,k_3,m_2,m_3,l_2,l_3)$-equivalence $G$ of $\stab_2$ and $\stab_3$ we can define
an isomorphism $G\circ F$ from $\stab_1$ to $\stab_3$
\begin{align}\label{compose}
G\circ F:= & (\coarse {l_2} G\otimes \idtrivial^{\otimes t_2})\circ S\circ(\coarse {l_2'} F \otimes \idtrivial^{\otimes t_1})
\end{align}
where $t_1=k_2'l_2+m_2'l_2,$ $t_2=k_2l_2'+m_2l_2'$, $\idtrivial$ is the identity morphism $\morphism {\idtrivial} \paulitrivial\paulitrivial$
and $S$ swaps the last $t_2$ qubit labels with the previous $t_1$ qubit labels involved in the trivial codes. 
This is a $(k_1',k_3',m_1',m_3',l_1',l_3')$-equivalence with
\begin{alignat}{3}
k_1'&:=k_1l_2'+k_2'l_2, 
&\qquad 
l_1'&:=l_1l_2',
&\qquad 
m_1'&:=m_1l_2'+m_2'l_2, 
\nl
k_3'&:=k_3l_2+k_2 l_2',
&\qquad 
l_3'&:=l_3l_2.
&\qquad 
m_3'&:=m_3l_2+m_2 l_2',
\end{alignat}

\section{Constraints and independent generators}

This section shows how a LTI set of independent generators can be constructed for any LPG that is the centralizer of a LPG.
$\PA=\sget{\PA_g}$, $\PB=\sget{\PB_g}$ are LPGs on a given 2D lattice.
$\pauli_g$ denotes the set of single qubit $X$, $Z$ Pauli operators.

\subsection{Constraints}

\begin{figure}
\psfrag{n}{$n$}
\psfrag{m}{$m$}
\psfrag{l}{$l$}
\begin{center}
\includegraphics[width=.45\columnwidth]{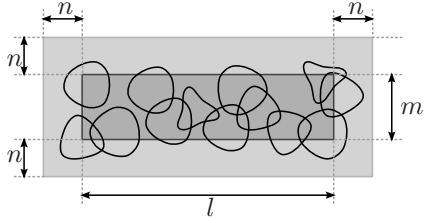}
\end{center}
\caption{
Illustration of \propref{long_is_narrow}.
$\gamma_l$ appears in darker grey.
The closed shapes represents the support of elements of $\negr {\PB_g}{\phi}$.
The support of $S$ is contained in the larger rectangle.
}
\label{fig:long_narrow}
\end{figure}

\begin{defn}
The groups of $\PB$-constraints of $\PA_g$ are
\begin{align}
\cnstrr{\PB}{\PA_g}&:=\set {S\subset \PA_g-\phase}{\syndrset{\PB} S=1},\nl
\cnstrfinr{\PB}{\PA_g}&:=\powersetfin{\PA_g}\cap\cnstrr{\PB}{\PA_g}
\end{align}
\end{defn}

When $\PB=\pauli$ we drop the label.
These constraints give us a measure of how non-independent the elements of $\PA_g$ are with respect to their action in $\PB$.
$\PA_g$ is independent iff $\cnstrfin{\PA_g}=\sset \emptyset.$

The following proposition, illustrated in \fref{long_narrow}, states that given $a\in\PA$ such that it commutes with all generators of $\PB$ outside of a block of a fixed width but arbitrary length, then there exist another $a'\in\PA$ that anticommutes with the same elements of $\PB$ as $a$ does but such that it is the product of a set of generators of $\PA$ inside a slightly larger block.

\begin{prop}\label{prop:long_is_narrow}
For any $m\in\N^\ast$ there exists $n\in\N_0$ such that for every $a\in\PA$ and $l\in\N_0$ with
\be
\negr {\PB_g}a
\subset \PB_g\overlap{\gamma_l}, \qquad \gamma_l:=\bl 0 0 {l} {m},
\ee
there exists $S\in\powersetfin{\PA_g}$ with 
\be
\syndrset \PB S = \syndrset \PB a,\qquad \supp S\subset \thk^n (\gamma_l).
\ee
We denote by $N(\PA_g,\PB_g,m)\in\N_0$ the smallest of such $n$.
\end{prop}

\begin{proof}
Let $\PB_1:=\sget{\PB_g\overlap{\gamma_1}}$. 
For $\phi\in \morph {\PB_1}$ let 
\be
A_\phi :=\set{a\in\PA}{
\negr {\PB_g} a
\subset\PB_g\overlap{\gamma_\infty}, \syndr {\PB_1} a=\phi}
\ee
and $\Phi_1:=\set{\phi\in \Phi_{\PB_1}}{A_\phi\neq \emptyset}$.
$\Phi_1$ is finite because $\PB_1$ is finite.
For each $\phi\in \Phi_1$ we choose $a_\phi\in A_\phi$ that satisfies  
$\negr {\PB_g} {a_\phi}\subset \PB_g\overlap{\gamma_l}$
for a minimal $l$ and let $l_\phi$ be this minimum value.
We also choose $S_\phi\in\powersetfin{\PA_g}$ with $\syndr \PB {a_\phi} = \syndrset \PB {S_\phi}$.
Finally, we choose $n\in\N_0$ so that $\supp {S_\phi}\subset \thk^{n}({\gamma_{l_\phi}})$ for any $\phi\in \Phi_1$.

We now prove the proposition by induction on $l$. 
Let $\phi=\syndr{\PB_1}a$, so that $a\in A_\phi$.
If $l=1$, $S=S_\phi$ satisfies the required properties because $\negr {\PB_g} {a_\phi}=\negr {\PB_1} {a_\phi}=\negr {\PB_1} {a}=\negr {\PB_g} {a}$. 
If $l>1$, let $a' = \trans {-1} 0 (aa_\phi)$. 
Since $l_\phi\leq l$, $\negr {\PB_g}{a'}\subset \PB_g\overlap{\gamma_{l-1}}$ and by induction there exists $S'\in\powersetfin{\PA_g}$ with 
$\syndrset \PB {S'} = \syndrset \PB {a'}$ and $\supp {S'}\subset \thk^n (\gamma_{l-1}).$
Then $S=S_\phi\trans 1 0 (S')$ satisfies the required properties.
\end{proof}

It is worth rewriting this in the special case $\PB=\pauli$.

\begin{cor}\label{cor:long_is_narrow_pauli}
For any $m\in\N^\ast$ there exists $n\in\N_0$ such that for every $a\in\PA$ and $l\in\N_0$ with
\be
\supp a\subset \gamma_l:=\bl 0 0 {m} {l},
\ee
there exists $S\in\powersetfin{\PA_g}$ with $a \in\pro S$ and
\be
\supp S\subset \thk^n (\gamma_l).
\ee
We denote by $N(\PA_g,m)\in\N_0$ the smallest of such $n$.
\end{cor}

\begin{lem}\label{lem:bounded_cnstr}
$\cnstrfinr {\PB}{\PA_g}$ admits a LTI set of generators.
\end{lem}

\begin{proof}
Let $n_1=N(\PA_g,\PB_g,1)$ according to \propref{long_is_narrow} and $n_2=N(\PA_g,\PB_g,1)$ according to the axis-exchanged version of \propref{long_is_narrow}.
Set $N=\max(n_1,n_2)$.
We claim that
\be
G:=\set{S\in\cnstrfinr {\PB}{\PA_g}}{\range  S\leq L},
\ee
where $L:=8N+2$, generates $\cnstrfinr {\PB}{\PA_g}$.

Take any $S\in \cnstrfinr{\PB}{\PA_g}$.
If $\range S\leq L$ then $S\in\sget G$. 
In other case, without lost of generality we assume $\supp R \subset \bl 0 0 a b$ with $a>b$ and $a> L$. 
Let $l= 4N+1$ and set $S=S_1\add S_2$ with $S_1=S|_{\bl 0 0 {l} {b}}$. 
Then $\supp {S_2}\subset\bl {l-1} 0 {a}{b}$, $\syndrset \PB {S_1} = \syndrset \PB {S_2}=:\phi\in\syndri \PB \PA$, and 
\be
\suppr {\PB_g} \phi\subset \suppr {\PB_g} {S_1} \cap \suppr {\PB_g}{S_2} \subset \PB_g\overlap{\bl {l-1}{0}{l}{b}}.
\ee
Due to \propref{ } (up to a translation), there exists $S_3\in \powersetfin{\PA_g}$ with $\phi=\syndrset \PB {S_3}$ and $\supp {S_3}\subset \thk^N (\bl {l-1}{0}{l}{b}).$
Then $S=S_1'+S_2'$, where $S_1'=S_1\add S_3$ and $S_2'=S_2\add S_3$ are elements of $\cnstrfinr {\PB}{\PA_g}$.
But $S_i'\in \bl {u_i+a_i} {v_i+b_i}{a_i} {b_i}$, $i=1,2$, with $a_1=l+N$, $b_1=b+2N$, $a_2=a+1+N$ and $b_2=b+2N$.
Thus $a_i+b_i<a+b$ and by induction on $a+b$ it follows that $R\in\sget G$.
 \end{proof}

Since $\cen\PB\PA = \PA\cap \bigcup\proi {\cnstrfinr \PB{\PA_g}}$, we have:
\begin{cor}\label{cor:centralizer}
$\cen \PB \PA$ is a LPG.
\end{cor}

When $a\in\cen\PB\pauli$ does not have support on the boundary of a region, the restriction of $a$ to this region still belongs to the centralizer:

\begin{prop}\label{prop:separate}
Let $\PA=\cen\PB\pauli$.
For any $a\in\PA$ and $\gamma\in\powersetfin \Sigma$ 
\be
\supp a\cap(\thk^1(\gamma)-\gamma)=\emptyset \quad\Longrightarrow\quad a|_\gamma\in\phase \PA.
\ee
\end{prop}

\begin{proof}
Since $\PB_g$ is 2-bounded, for any $b\in\PB_g$ with $\supp b\cap\gamma\neq\emptyset$ we have $\supp b\cap(\supp a-\gamma)=\emptyset$ and thus $\comm b {a|_\gamma} = \comm b a=1$.
\end{proof}

\begin{lem}\label{lem:global_constraints}
Let $\PA=\cen\PB\pauli$ and $\PA_g$ be independent. 
$\cnstr {\PA_g}$ is finite.
\end{lem}

\begin{proof}
Let $n_1:=N(\PA_g,2)$ according to \corref{long_is_narrow_pauli} and $n_2:=N(\PA_g,2)$ according to the axis-exchanged version of \corref{long_is_narrow_pauli}.
Set $N=\max(n_1,n_2)$ and $D=2N+2$.
First, we want to show that for any $R\in\cnstr{\PA_g}$ we have
\be
R\overlap{\bl 0 0 2 {D}}=\emptyset\quad\Longrightarrow \quad R\overlap{\bl 2 N \infty {D-N}}=\emptyset.
\ee
We start observing that for any $l\in\N^\ast-\sset{1}$, there exist $k_1,k_2\in \N^\ast$ with $k_2-2>k_1>l$ and such that
\be
\left(\trans {k_2-k_1} 0(R) \add R\right)\overlap{\bl {k_2} {0} {k_2+2} {D}}=\emptyset,
\ee
because $\set{S\subset \PA_g}{S=S\overlap{\bl 0 0 2 {D}}}$ is finite. Let
$R_{0}:=R\overlap {\bl {2}{0}{\infty}{D}},$ $R_1 :=\trans {k_2-k_1} 0 (R_0)\add R_0.
$
Clearly we have $R_1\overlap {\bl {0}{0}{2}{D}\cup\bl {k_2}{0}{k_2+2}{D}}=\emptyset.$
Then $R_2:=R_1\overlap {\bl 2 0 {k_2} D}$ satisfies $\supp{\pro {R_2}} \subset \gamma_1\cap \gamma_2$ for $\gamma_1:=\bl 0 {-2} {k_2+2} {0}$, $\gamma_2:=\bl 0 {D} {k_2+2} {D+2}.$
According to \propref{separate} there exist $a_i\in\PA$, $i=1,2$, such that $a_1a_2\in\pro {R_2}$.
Applying \corref{long_is_narrow_pauli}, there exist $S_i\in \powersetfin{\PA_g}$, $i=1,2$ with $a_i\in\pro {S_i}$ and $\supp {S_i}\subset \thk^N (\gamma_i)$. 
Then $R_3=R_2\add S_1\add S_2\in\cnstrfin{\PA_g}=\sset\emptyset$ and $R_1\overlap{\bl 2 {N} l {D-N}}=R_3\overlap{\bl 2 {N} l {D-N}}=\emptyset$.
This can only be true if $R\overlap{\bl 2 {N} l {D-N}}=\emptyset$.

We next show that for any $R_1,R_2\in\cnstr{\PA_g}$ with $R_1\overlap{\bl 0 0 2 D}=R_2\overlap{\bl 0 0 2 D},$
we have $R_1=R_2$, which implies the finiteness of $\cnstr{\PA_g}$. 
Let $R=R_1\add R_2$.
Then according to the above result $R\overlap{\bl 2 {N} \infty {D-N}}=\emptyset$. 
Moreover, we can apply the result with the first coordinate inverted and a suitable translation, getting $R\overlap{\bl {-\infty} {N} \infty {D-N}}=\emptyset$.
But now for any $j\in\Z$ we have $R\overlap {\bl j {N} {j+D} {2+N}}=\emptyset$.
Applying the same reasoning with the coordinates exchanged gives $R\overlap {\bl {j+N} {-\infty} {j+D-N} \infty}=\emptyset$ for any $j\in\Z$. 
Thus $R_1\add R_2=\emptyset$.
\end{proof}

\subsection{Independent generators}

\begin{figure}
\psfrag{nn}{$2n$}
\psfrag{m}{$m$}
\psfrag{l}{$2l$}
\psfrag{mm}{$M$}
\begin{center}
\includegraphics[width=.45\columnwidth]{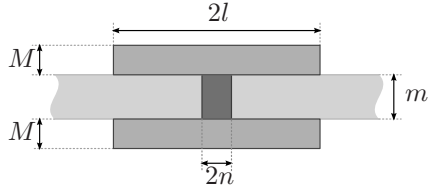}
\end{center}
\caption{
Illustration of \propref{extend_narrowing}.
$R$ has no support on the region $\gamma$, in medium grey.
The support of $S$ is contained in the light grey area, which extends to infinity on both sides.
$S$ and $R$ coincide in the region $\gamma_{-n}^n$, in dark grey.
}
\label{fig:extend_narrow}
\end{figure}

We need a couple of lemmas as a preparation for the main result of this section. The first states that given a constraint $R$ without support on the `sides' of a `corridor', as long as this sides are wide enough and the corridor long enough, there will be another constraint $S$ that coincides with $R$ in the centre of the corridor but that has support only on an infinite version of the corridor, as illustrated in \fref{extend_narrow}.

\begin{prop}\label{prop:extend_narrowing}
Let $\PA_g$ be $M$-bounded.
For any $m,n\in\N_0$ there exists $l\in \N_0$ such that the following holds.
For any $R\in \cnstr{\PA_g}$ satisfying
\be\label{cond_R}
R\overlap{\gamma}=\emptyset,\qquad \gamma:=\bl {-l} {-M} {l} {0}\,\cup\, \bl {-l} {m} {l} {m+M},
\ee
there exists $S\in\cnstr{\PA_g}$ such that
\be\label{prop_S}
 S\subset {\PA_g|_{\gamma_{-\infty}^{\infty}}},\quad (S\add R)\overlap{\gamma_{-n}^{n}}=\emptyset,\quad \gamma_{a}^{b}:=\bl a 0 {b}{m}.
\ee
We denote by $L(\PA_g,m,n)$ the minimal such $l$.
\end{prop}

\begin{proof}
For $a<b$, we set $\gamma_b^a:=\gamma_a^b$.
Let $l\in\N_0$ be such that for any $R\in \cnstr{\PA_g}$ satisfying \eqref{cond_R}
there exist $m_x,n_x\in\N_0$, $x=\pm$, with $n\leq xm_x<xm_x+M\leq xn_x\leq r-M$ and such that
\be
\left(\trans {d_x} 0(R) \add R\right)\overlap{\gamma_{n_x}^{n_x+xM}}=\emptyset,\qquad d_x:=n_x-m_x.
\ee
There exists such $l$ because the set $\PA_g\overlap{\gamma_0^{M}}$ is finite and $\PA_g\overlap{\gamma_t^{M}}=\trans t 0 (\PA_g\overlap{\gamma_0^M})$.
We claim that this choice for $l$ satisfies the statement in the lemma.
To check it, given $R$, $m_x$, $n_x$ as above let 
\begin{align}
R_x&:=R\overlap{\gamma_{m_x}^{n_x}}-R\overlap{\gamma_{m_x-xM}^{m_x}},\quad
R_x^\infty:=\bigcup_{a\in\N^\ast}\trans {ad_x} 0 (R_x),\nl
S&:=R\overlap{\gamma_{n_{-}}^{n_+}} \cup R_{-}^\infty \cup R_+^\infty.
\end{align}
Due to \eqref{cond_R}, $R\overlap{\gamma_{n_{-}}^{n_+}}\subset R|_{\thk^M(\gamma_{n_{-}}^{n_+})}\subset R|_{\gamma\cup\gamma_{-r}^{r}}\subset{\PA_g|_{\gamma_{-\infty}^{\infty}}}$ and thus $S\subset {\PA_g|_{\gamma_{-\infty}^{\infty}}}$. 
The definition of $S$ ensures that
\be\label{propS1}
(S\add R)\overlap{\gamma_{n_--M}^{n_++M}}=\emptyset,
\ee
and thus in particular $(S\add R)\overlap{\gamma_{-n}^n}=\emptyset$.
The definition is also such that for any $a\in\N^\ast$, $x=\pm$,
\be\label{propS2}
(\trans {ad_x} 0 (S)\add S)\overlap{\gamma_{m_x+ad_x}^{\infty}}=\emptyset
\ee
To complete the proof, it suffices to show that for any $p\in\pauli_g$ we have $\comm S p=1$.
Let $\supp p\subset\bl t {-\infty} {t+1}{\infty}$ for some $t\in\N_0$.
W.l.o.g. we set $t\geq0$.
If $t\leq n_+$, $\comm S p=\comm R p = 1$ according to \eqref{propS1}.
In other case $t=m_+ + ad_+ +t'$ for some $a\in\N^\ast$, $0\leq t'< d_+$.
Using \eqref{propS2} we get $\comm S p = \comm {\trans {ad_+} 0 (S)} p=\comm S {\trans {-ad_+} 0 (p)}=1$.
\end{proof}

Next we give a procedure to remove dependent generators of $\PA_g$ from selected regions preserving translational invariance.
The constraints must have generators that do not have support on more than one of these regions at a time.

\begin{prop}\label{prop:erase_constraints}
Let $\PA_g$ have period $d$ and $\cnstrfin{\PA_g}=\sget{R_0}$ for some LTI set $R_0\subset\powersetfin{\PA_g}$.
Let $\gamma\in\powersetfin\Sigma$ be such that for any $S\in R_0$ there is at most  a pair $(m,n)\in\Z^2$ such that
\be
S\overlap {\trans {md}{nd} (\gamma)} \neq \emptyset.
\ee
Then there exists $A\subset \PA_g\overlap\gamma$ such that 
\be
\PA=\sget{\PA_g'},\qquad
\PA_g':=\PA_g-\strans d (A),
\ee
and $S\overlap{\strans d (\gamma)}=\emptyset$ for any constraint $S\in\cnstrfin{\PA_g'}$.
\end{prop}

\begin{proof}
If $S\overlap{\gamma}=\emptyset$ for any constraint $S\in\cnstrfin{\PA_g}$, clearly $A=\emptyset$ suffices. 
In other case, choose $a\in \PA_g\overlap\gamma$ and $S\in R_0$ such that $a\in S$.
We claim that $\PA=\sget{\PA_g'}$ with $\PA_g':=\PA_g-\strans d (a)$.
Due to translational invariance, it suffices to show $a\in\sget {\PA_g'}$. 
But $a\in\pro{S-\sset{a}}$, and since $\strans d (a) \cap S=\sset a$ we have $(S-\sset{a})\subset \PA_g'$.
Let $S_{m,n}:=\trans {md}{nd} (S)$. Given $S'\in R_0$, let $f(S')=S'+S_{m,n}$ if $\trans{md}{nd}(a)\in S'$, $f(S')=S'$ otherwise. 
Then the set $R_0'=f[R_0]$ generates $\cnstrfin{\PA_0'}$.
Moreover, $\PA_g'$ and $R_0'$ satisfy the conditions of the lemma, because $\supp {S\add S_{m,n}}\subset \supp S\cup\supp {S_{m,n}}$.
Since $|(\PA_g'\overlap\gamma)|=|(\PA_g\overlap\gamma)|-1$, the result follows by induction on $|(\PA_g\overlap\gamma)|$.
\end{proof}

\begin{figure}
\psfrag{2N}{$2N$}
\psfrag{L}{$L$}
\psfrag{K}{$K$}
\psfrag{D}{$D$}
\psfrag{G0}{$\gamma_0$}
\psfrag{G1}{$\gamma_1$}
\psfrag{G2}{$\gamma_2$}
\psfrag{G1X}{$\gamma_1^X$}
\psfrag{G2X}{$\gamma_2^X$}
\psfrag{G3X}{$\gamma_3^X$}
\psfrag{G4X}{$\gamma_4^X$}
\psfrag{G5X}{$\gamma_5^X$}
\begin{center}
\includegraphics[width=.9\columnwidth]{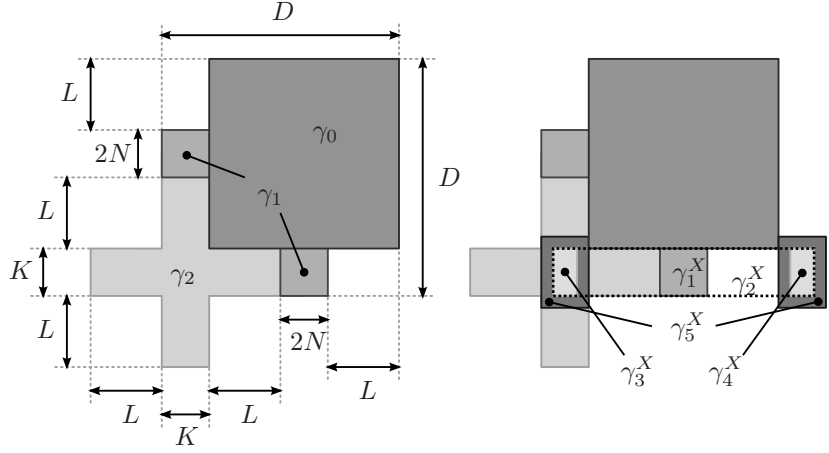}
\end{center}
\caption{
The regions involved in the proof of \thmref{independent}. $\gamma_2^X$ is outlined with dashed lines.
}
\label{fig:independent}
\end{figure}

\begin{thm}\label{thm:independent}
Let $\PA=\cen\PB\pauli$. Then $\PA$ admits LTI independent generators.
\end{thm}

\begin{proof}
We will construct from $\PA_g$ a LTI set of independent generators $\PA_3$ through a series of intermediate sets $\PA_i$ so that $\PA_3\subset\PA_2\subset\PA_1\subset\PA_0:=\PA_g$.
At each step we will apply \propref{erase_constraints} to a given region $\gamma_i$ and obtain 
\be\label{step}
\PA_{i+1}=\PA_i-\strans D (A_i)
\ee
for a suitable $D\in\N^\ast$, $A_i\subset\PA_i\overlap{\gamma_i}$.
Let us define $R_i:=\cnstrfin{\PA_i}$.
Since we apply \propref{erase_constraints}, we have $S\overlap {\strans D (\gamma_i)}=\emptyset$ for any constraint $S\in R_j$ if $j> i$.
The goal is thus to have $\sset\emptyset=R_3\subset R_2\subset R_1\subset R_0.$
$R^i_g$ will denote a LTI set of generators of $R_i$ of period $D$.

\noindent {\bf 0)}
W.l.o.g. we assume that $\PA^0_g$ and $R_g^0$ have period $1$ and are $2$-bounded.
Let $n_1:=N(\PA_g,2)$ according to \corref{long_is_narrow_pauli} and $n_2:=N(\PA_g,2)$ according to the axis-exchanged version of \corref{long_is_narrow_pauli}.
Set $N=\max(1,n_1,n_2)$ and $K:=2N+2$. 
Let $l_1=L(\PA_g,K,2)$ according to \propref{extend_narrowing} and $l_2=L(\PA_g,K,2)$ according to the axis-exchanged version of \propref{extend_narrowing}.
Set $L:=\max(l_1,l_2)$ and $D:=K+2L+2N$.
The geometry of the proof, see \fref{independent}, involves the regions
\begin{align}
\gamma_0&:=\bl K K D D,\quad &
\gamma_1&:=\gamma_1^X\cup\gamma_1^Y,\quad&
\gamma_2&:=\bl {2-L} 0 {K+L-2} K\cup\bl 0 {2-L} K {K+L-2},\nl
\gamma_1^X&:=\bl {K+L-2}  0 {K+L+2} {K}, \quad &
\gamma_1^Y&:=\bl   0 {K+L-2}  K {K+L+2},\quad&
\gamma_2^X&:=\bl {N}  0 {D+K-N} {K}, \nl
\gamma_2^Y&:=\bl  0 {N}  K {D+K-N},\quad&
\gamma_3^X&:=\bl {N} {0} {K-N} {K}, \quad &
\gamma_3^Y&:=\bl {0} {N} {K} {K-N},\nl
\gamma_4^X&:=\trans D 0 (\gamma_3^X), \quad &
\gamma_4^Y&:=\trans D 0 (\gamma_3^X),\quad&
\gamma_5^X&:=\thk^N(\gamma_3^X\cup\gamma_4^X), \nl
\gamma_5^Y&:=\thk^N(\gamma_3^Y\cup\gamma_4^Y),\quad&
\gamma_6^X&:=\gamma_2^X\cup\gamma_5^X, \quad &
\gamma_6^Y&:=\gamma_2^Y\cup\gamma_5^Y.
\end{align}

\noindent {\bf 1)}
$\PA_0$, $R_g^0$ and $\gamma_0$ satisfy the conditions of \propref{erase_constraints} and we get $\PA_1$ setting $i=0$ in \eqref{step}.
Let for $\alpha=X,Y$
\be
Q^\alpha:=\set{S\in R_1}{S\overlap {\gamma_1^\alpha}\neq\emptyset, S|_{\gamma_6^\alpha}\neq S}.
\ee
We claim that $R_1=\sget{R^1_g}$ with
$R^1_g:=R^1-\strans D (Q^X\cup Q^Y).$
Take any $R\in R_1$. 
The set
\be
\set{\gamma\in \strans D (\sset{\gamma_1^X,\gamma_1^Y})}{R\overlap\gamma\neq\emptyset}
\ee
is finite, and thus we proceed by induction on its cardinality.
If it is empty $R\in R_1$. In other case, due to translation and axis-exchange symmetries we can assume that $S\overlap {\gamma_1^X}\neq\emptyset$.
Then it suffices to show that there exists $R'\in R_1$ such that  $(R\add R')\overlap{\gamma_1^X}=\emptyset$.
To construct $R'$ we first notice that $R\in R_0$ and apply \propref{extend_narrowing} ---up to a translation $\trans {K+L} 0$--- to obtain $S\in R_0$ such that
\be
 S\subset {\PA^0_g|_{\gamma_{-\infty}^{\infty}}},\quad (S\add R)\overlap{\gamma_{-M}^{M}}=\emptyset,\quad \gamma_{a}^{b}:=\bl {K+L+a} 0 {K+L+b}{K}.
\ee
Setting $S':=S|_{\gamma_2^X}$ we get $(S' \add R)\overlap{\gamma_{-M}^{M}}=\emptyset$ and 
$\supp{S'}\subset \gamma_3^X\cup\gamma_4^X.$
From \propref{separate}, \corref{long_is_narrow_pauli} and the involved geometry it follows the existence of $S_3,S_4\subset\powersetfin{\PA_0}$ such that $\supp {S_i}\subset \thk^N {\gamma_i^X},$ $i=3,4,$
and $S'\add S_3\add S_4\in R_0$. We can take $R':=S' \add S_1\add S_2$ because $R'|_{\gamma_6^X}=R'$ and $(R\add R')\overlap{\gamma_{-M}^{M}}=\emptyset$.

\noindent {\bf 2)}
$\PA_1$, $R_g^1$ and $\gamma_1$ satisfy the conditions of \propref{erase_constraints} and we get $\PA_1$ setting $i=1$ in \eqref{step}.
We claim that $R^2_g:=\strans D \set{S\in R_2}{\supp S\subset \gamma_2}$
generates $R_2$. Indeed, 
\be\label{total}
\Sigma=\strans D (\gamma_0\cup\gamma_1\cup\gamma_2),
\ee
gives $\supp S\subset\strans D (\gamma_2)$ for any $S\in R_2$.
Then $R_2\ni S_{m,n}:=S|_{\trans m n (\gamma_2)},$
because for any $b\in\PB_g$ there is at most one pair $(m,n)\in\Z^2$ with $\supp b\cap \trans m n (\gamma_2)\neq\emptyset$.

\noindent {\bf 3)}
$\PA_3$, $R_g^2$ and $\gamma_2$ satisfy the conditions of \propref{erase_constraints} and we get $\PA_3$ setting $i=2$ in \eqref{step}.
Due to \eqref{total} we have $R_4=\sset\emptyset$ as desired.
\end{proof}

\section{Charge and strings}\label{sec:charge_cnstr}

This section defines charges, establishes their duality with constraints and introduces strings.

\subsection{Charge groups}

Let $\PA$ be a LPG with independent generators $\PA_g$ of a LPG.
We can define
\be
\morphfin \PA:=\morphfin{\PA_g}
\ee
because the choice of the generating set $\PA_g$ is immaterial.
\begin{defn}
The charge group of $\PA$ is
\be\label{charge_group}
 \charge\PA:=\frac{\morphfin \PA}{\syndri{\PA}\pauli}.
\ee
Its elements are the charges of $\PA$.
\end{defn}
When $\phi_1,\phi_2\in c\in \charge\PA$ we write $\phi_1\sim\phi_2$ and $\chg \phi=c$.
A LPG isomorphism  $F$ from $\PA\subset\pauli_A$ to $\PB\subset\pauli_B$ induces an isomorphism
\functmap {F^\ast} {\morphfin {\PA}}{\morphfin\PB}{\phi}{\phi\circ F\inv}
that maps charges to charges because $F^\ast [\syndri\PA{\pauli_A}]=\syndri\PB{\pauli_B},$ giving rise to an isomorphism
\be
\funct {\charge F} {\charge {\PA}}{\charge\PB.}
\ee
Coarse graining does not affect the group of charges, in the sense that there exists a natural isomorphism
\be\label{coarse_natural}
\charge\PA\simeq \charge{\coarse l \PA}
\ee
for any $l\in\N^\ast$.
LPG composition gives yet another natural isomorphism
\be\label{composition_natural}
\charge{\PA\otimes\PB}\simeq\charge{\PA}\times\charge\PB.
\ee

\subsection{Charge of generators}

\begin{defn}
We say that $a\in\PA_g-\phase$ has charge $c\in\charge\PA$ if $\mor{\PA_g} {\sset{a}}\in c.$
\end{defn}

The charge of $a$, denoted $\chgr {\PA_g} a$, depends on the whole $\PA_g$ and
\be\label{charge_composition}
\chg \phi = \prod_{a\in\negr{\PA_g}\phi} \chgr{\PA_g} a.
\ee
We need a prescription that tell us how the charges of generators change as the generating set changes preserving translational invariance.
\begin{prop}\label{prop:charge_change}
Let $\PA$ be a LPG with a LTI set of independent generators $\PA_g$ of period $L$.
Let $a,b\in\PA_g-\phase$ satisfy $b\not\in\strans L (\sset{a})$.
Then
\be
\PA_g':=\PA_g\cup\strans L (\sset{ab})-\strans L (\sset{a})
\ee
is a LTI set of independent generators of $\PA$ and
\be\label{charge_change}
\chgr {\PA_g'} {ab} = \chgr{\PA_g} a, \qquad
\chgr {\PA_g'} {b} = \chgr{\PA_g} a \chgr{\PA_g} b.
\ee
\end{prop}

\begin{proof}
Since $ab\not\in\PA_g$ because $\PA_g$ is independent, we have $b,ab\in\PA_g'$. 
Then $a\in\sget{\PA_g'}$ and by translational invariance $\PA=\sget{\PA_g'}$.
Given $S'\in\cnstrfin{\PA_g'}$, let
\be
z:=\set{(i,j)\in \Z^2}{\trans{iL}{jL} (ab)\in S'}.
\ee
Then we can construct $S\in \cnstrfin{\PA_g}$ setting
\be
S:=S'\add \sum_{(i,j)\in z} \trans{iL}{jL} (\sset {a,b})-\sum_{(i,j)\in z} \trans {iL}{jL} (\sset{ab}).
\ee
But $|S|\geq |S'|$ because $\strans L (\sset{a})\cap S'=\emptyset$, and thus $\PA_g'$ is independent.
Let
\be
\phi_a:=\mor{\PA_g} {\sset{a}}, \qquad \phi_b:=\mor{\PA_g} {\sset{b}}.
\ee 
To recover \eqref{charge_change}, apply \eqref{charge_composition} to both $\PA_g$ and $\PA_g'$ to get
\begin{align}
\chg {\phi_a}&=\chgr{\PA_g} a = \chgr{\PA_g'} {ab},\nl
\chg {\phi_b}&=\chgr{\PA_g} b = \chgr{\PA_g'} {ab}\chgr{\PA_g'} {b}.
\end{align}
\end{proof}

\subsection{Charge-constraint duality}

Given $R\in \cnstr{\PA_g}$ and $c\in\charge \PA$ we can define
\be
c \cdot R := \phi(R),\qquad \phi \in c,
\ee
because the choice of $\phi$ is inmaterial.
\begin{prop}\label{prop:constraint_duality}
If $\cnstr{\PA_g}$ is finite, $\cnstr{\PA_g}$ and $\charge \PA$ are dual.
\end{prop}

\begin{proof}
We construct dual sets of generators for $\cnstr{\PA_g}$ and $\charge \PA$.
Let $C\subset\charge\PA$ be the set of charges $c$ such that $\chgr{\PA_g} {a_c}= c$ for some $a_c\in\PA_g$.
Clearly $\charge\PA=\sget {C}$ and we can choose some countable $C_g\subset C$ as an independent set of generators of $C$.
For each $c\in C_g$ define the set $R_c\subset\PA_g$ as follows
\begin{align}
R_c&:=\bigcup_{c'\in \bar c}{A_{c'}}, \qquad 
\bar c:=\charge\PA-\sget {C_g-\sset c},\nl
A_c&:=\set{a\in\PA_g-\phase}{\chgr{\PA_g}{a}= c}.
\end{align}
Checking that $R_c\in\cnstr{\PA_g}$ amounts to show that $\phi(R_c)=1$ for every $\phi\in\syndri {\PA_g}{\pauli}$.
But given $c_1,c_2\in \bar c$, $c_3,c_4\not\in\bar c$ we have $c_1c_2,c_3c_4\not\in\bar c$ and $c_1c_3\in\bar c$.
Then $\chg \phi =1\not\in \bar c$ gives $|R_c \cap \negr{\PA_g}\phi|$ and thus $\phi(R_c)=1$ as desired. 
Since $a_c\in R_{c'}$ iff $c=c'$, for $c,c'\in C_g$ we have $c\cdot R_{c'}=\mor{\PA_g}{\sset{a_c}}(R_{c'})=1-2\delta_{c,c'}$.
Thus $C_g$ is finite. 
To show that the $R_c$ generate $\cnstr{\PA_g}$, consider that $b\in R\in\cnstr{\PA_g}$ with $a_c\not\in R$ for every $c\in C_g$.
There exists $A\subset \set{a_c}{c\in C_g}$ such that $\chgr{\PA_g}b=\prod_{a\in A}\chgr{\PA_g}a$.
Thus $\phi(R)=-1$ for $\phi=\mor{\PA_g}{\sset b\cup A}\in\syndri{\PA}\pauli$, a contradiction, showing that $R$ is empty.
\end{proof}

\subsection{Coarse graining}

We next show that, given a collection of LPGs with independent generators and finite charge groups, by coarse graining the lattice it is possible to gain charge translational symmetry and other properties. First, there exist generators with any given charge ---among a generating set of charges--- and support in a single site. Second, given a Pauli operator $p$ that commutes with all generators in $\PA_g$ with support outside a connected region $\gamma$, there exist another operator $p'$ that anticommutes with the same generators as $p$ but has support in a region only slightly larger that $\gamma$ ---preparing the ground for string operators. Third, given an element of one of this LPGs with support in a block, its generators have support in a slightly larger block.

\begin{prop}\label{prop:coarse_grain2}
Consider a set of LPGs $\bar \PA_1,\dots,\bar \PA^n$, on a given qubit lattice, all admiting LTI sets of independent generators and with finite charge groups.
There exists $L\in\N^\ast$ such that the coarse grained LPGs 
\be
\PA^k=\coarse L {\bar \PA^k}\subset\pauli, \qquad k=1,\dots,n,
\ee
for any independent set of generators $\sset{c_l^k}_{l=1}^{m_k}$ of $\charge{\PA^k}$, $m_k\in \N_0$, 
admit LTI sets of independent generators $\PA^k_g$ such that the properties 1 and 2 in \propref{coarse_grain1} and the following ones are satisfied.
\begin{enumerate}
\setcounter{enumi}{2}
  \item 
  For any $c\in \charge {\PA^k}$ there exists $\phi\in c$ such that $\suppr{\PA_g^k} \phi\subset \sset \Site.$
  \item 
  $\phi\sim \trans i j (\phi)$ for any $\phi\in \morphfin {\PA^k}$ and $i,j\in\Z$.
  \item 
  For fixed $k$, for any $p\in\pauli$ and connected set of sites $\gamma\in\powersetfin\Sigma$ with $\suppr{\PA^k_g} p\subset \gamma$ there exists $p'\in\pauli$ with
  $\syndr {\PA^k}{pp'}=\emptyset$, $\supp {p'}\subset\thk^1 (\gamma).$
  \item For any $A\in\powersetfin{\PA_g^k}$ and $\gamma=\bl i j {i+L}{j+L}$, where $i,j\in\Z$ and $L\in\N_0$,
  \be
   \supp {\pro A} \subset \gamma\quad\Longrightarrow\quad \supp A \subset\thk^1 (\gamma).
  \ee
 \item
For every $l=1,\dots,m_k$ there exists $a_l^k\in\PA_g^k$ with charge $c_l^k$, $\supp{a_l^k}= \sset \Site$.
 \end{enumerate}
\end{prop}

\begin{proof}
All properties are preserved under coarse graining.
We show, for properties $x=3,\dots,7$, that if properties 1 to $x-1$ are satisfied then property $x$ is also satisfied by further coarse graining or changing the sets of generators.

\noindent {\bf 3)}
Since $\charge{\PA^k}$ is finite, there exist $m\in\N^\ast$ such that for any $k$ and $c\in \charge {{\PA^k}}$ there exists $\phi\in c$ with
$\suppr{\PA_g^k} \phi\subset \bl 0 0 {m}{m}.$
Then the LPGs $\coarse {m}{\PA^k}$ satisfy point 3.

\noindent {\bf 4)}
The key observation here is that for any $\phi,\phi'\in\morphfin {\PA^k}$, $i,j\in\Z$, we have
$\phi\sim\phi'$ iff $\trans{i}{j}(\phi)\sim\trans {i}{j} (\phi').$
For each $k$ and $c\in\charge{\PA^k}$, choose $\phi_c^k\in\morphfin {\PA^k}$ such that 
$\suppr{\PA_g^k} {\phi^k_c}\subset \sset\Site.$
Since $\charge{\PA^k}$ is finite, there exist $m_1,m_2\in\N^\ast$, $m_1<m_2$, such that $\trans {m_1}0(\phi^k_c)\sim\trans {m_2}0 (\phi^k_c)$ for any $k$, $c\in \charge {\PA^k}$ and $i,j\in\Z$.
Or equivalently, $\phi^k_c\sim\trans {m}0 (\phi^k_c)$ with $m=m_2-m_1$.
Then for any $\phi\in c\in\charge{\PA^k}$ we have $\phi\sim \phi_c^k\sim\trans {m} 0  (\phi_c)\sim \trans{m}0 (\phi).$
The same reasoning in the other axis gives $m'\in\N^\ast$, and the LPGs $\coarse {mm'}{\PA^k}=\sget{{\coarse{mm'}{\PA_g^k}}}$ satisfy properties 1-4.

\noindent {\bf 5)}
For each $k$ and for each $\phi\in\syndri{\PA^k}\pauli$ with $\suppr{\PA_g^k} \phi\subset \bl 0 0 3 3$, choose $m'\in\N_0$ such that there exists $p\in\pauli$ with 
$\syndr{\PA^k} p = \phi$ and $\supp p\subset \thk^{m'} (\bl 0 0 3 3).$
Let $m-1$ equal the largest of such $m'$. 
We will show that for any $k$, $\phi\in\syndri {\PA^k}\pauli$ and $\gamma\in\powersetfin\Sigma$ with $\suppr{\PA_g^k} \phi\subset \gamma$ and $\gamma$ connected, there exists $p\in\pauli$ with
$\syndr {\PA^k}p=\phi$ and $\supp p\subset\thk^{m} (\gamma)$.
Then the LPGs $\coarse {m}{\PA^k}=\sget{{\coarse{m}{\PA_g^k}}}$ satisfy properties 1-5.
Let $\sigma\in\Sigma$ be such that the set $\gamma':=\gamma-\sset{\sigma}$ is connected.
There is always such a site if $\gamma\neq\emptyset$.
If $|\gamma|=1$, just notice that $\range {\thk^1 (\gamma)}=3$ and
$\thk^{m}(\gamma)=\thk^{m-1}(\thk^1 (\gamma)).$
In other case, choose $\sigma'\in\Sigma$ adjacent to $\sigma$.
Set 
 \be
 \phi_0:=\mor{\PA_g^k}{\negr{\PA_g^k}\phi \overlap {\sset{\sigma}}}.
 \ee
Due to the second property of \propref{coarse_grain1} $\suppr{\PA_g^k}{\phi_0}\subset\gamma_0:=\thk^1 (\sset\sigma)$.
 Choose $\phi_1\in \morphfin{\coarse {l_2}{\PA^k}}$ and $p\in\pauli$ with $\phi_0\sim\phi_1$ and
 \be
 \suppr {\PA_g^k} {\phi_1}\subset\sset{\sigma'},\qquad
 \syndr{\PA^k}p=\phi_0\phi_1,\qquad
 \supp p \subset \thk^{m-1}(\gamma_0).
 \ee
This is always possible because $\sigma'\in\thk^1(\sset\sigma)$.
 The result follows by induction on $|\gamma|$ observing that
\be 
\suppr{\PA^k_g}{\phi\phi_0\phi_1}\subset\gamma',\qquad
\thk^{m-1}(\gamma_0)\cup  \thk^{m}(\gamma')=\thk^{m}(\gamma).
\ee

\noindent {\bf 6)}
We will show that for any $k$, $A\in\powersetfin{\PA_g^k}$ and $\gamma=\bl i j {i+L}{j+L}$, where $i,j\in\Z$ and $L\in\N_0$,
  \be
   \supp {\pro A} \subset \gamma\quad\Longrightarrow\quad \supp A \subset\thk^2 (\gamma).
  \ee
Then the LPGs $\coarse {2}{\PA^k}=\sget{{\coarse{2}{\PA_g^k}}}$ satisfy properties 1-6.
Let $a\in A$ satisfy $\supp a\not\subset\thk^2(\bl{i}{j}{i+L}{j+L}).$
W.l.o.g, we assume that $\supp a\subset \bl u v {u+2}{v+2}$ with $u\geq i+1+L$.
Choose $w\in\N^\ast$ such that $\trans w 0 (a)\not\in A$ and $p\in\pauli$ such that
$\supp p \subset \thk^1 (\bl u v {u+w+2} {v+2})$, $\syndr {\PA^k} p = \phi\,\trans w 0 (\phi)$, $\phi:=\mor{\PA^k_g} {\sset a}.$
Then $\supp p\cap\supp {\pro A}=\emptyset$ implies $\comm p A=1$, but at the same time $\comm p A=\phi(A)=-1$, a contradiction.

\noindent {\bf 7)}
For each $k$, let us show that there exist a list of LTI sets of independent generators $\PA_g^k=\PA_g^{k,0}, \PA_g^{k,1},\dots, \PA_g^{k,m_k}=\tilde \PA_g^k$
such that for any $t=1,\dots,m_k$ and any $l=1,\dots, t$ there exists $a_l\in\PA_g^{k,t}$ with charge $c_l^k$ and satisfying \eqref{suppakl}.
Moreover, we do it in such a way that all these generating sets preserve properties 1-6, so that the final generating sets $\tilde\PA_g^k$ satisfy all the desired properties.
We construct $\PA_g^{k,t}$ from $\PA_g^{k,t-1}$ as follows.
There exists $\phi\in c_t^k$ with $\suppr{\PA_g^{k,t-1}} \phi = \sset\Site$.
Also, there exists $\hat a_0\in \negr {\PA_g^{k,t-1}} \phi$ such that $\hat a_0\neq a_k$ for any $1\leq k<t$, because the charge $c_t^k$ is independent of the charges $c_1^k,\dots,c_{t-1}^k$.
Label the elements of $\negr {\stab_g^{k,t-1}}\phi-\sset{\hat a_0}$ as $\sset{\hat a_i}_{i=1}^r$ and apply \propref{charge_change} repeatedly to perform the substitutions 
$\hat a_i\rightarrow \hat a_i'=\hat a_0\hat a_i$, $i=1,\dots, r$.
Let $\PA_g^{k,t}$ be the resulting set of generators.
Then $\chgr{\PA_g^{k,t}}{\hat a_0}=c_t^k$ and the rest of generators with support in $\Site$, updated or not, preserve their charge.
\end{proof}

\subsection{Strings}

Let $\PA$ satisfy the properties in \propref{coarse_grain2}.
We define the following sets of `string operators'.
\begin{defn}
Given $\phi\in\morphfin \PA$, $\gamma\subset\Sigma$, we set
\begin{align}
\str {\phi}{\gamma}:= \set{p\in\pauli}{\syndr{\PA_g}p=\phi, \supp p\subset\thk^1 (\gamma)}.
\end{align}
Given $\phi_1,\phi_2\in\morphfin\PA$  with $\suppr{\PA_g}{\phi_i}=\sset{\sigma_i}\subset\Sigma$, we set
\begin{align}
\strd {\phi_1}{\phi_2}:= \str {\phi_1\phi_2}{\pth{\sigma_1,\sigma_2}},&&\str 1 1:=\sset \id,\nl
\strd 1{\phi_1}:=\str {\phi_1}1:= \str{\phi_1}{\sset {\sigma_1}}.&
\end{align}
Finally, we set for $c\in\charge\PA$ and a path $\gamma:=(\sigma_i)_{i=1}^n$,
\begin{align}
\str {c}{\gamma}:=\bigcup \set{\str{\phi_1\phi_2}\gamma}{\phi_1,\phi_2\in c, \supp {\phi_1}\subset{\sigma_1}, \supp {\phi_2}\subset{\sigma_n}}.
\end{align}
\end{defn}

In general $\str\phi\gamma$ can be empty.
However, suppose that for some $n\in\N_0$
\be
\phi=\prod_{k=1}^n\phi_k,\quad\phi_k\in\syndri{\PA}\pauli,\qquad
\gamma=\bigcup_{k=1}^n\gamma_k,\quad \gamma_k\in\powersetfin\Sigma,
\ee
with the sets $\gamma_k$ connected and $\suppr{\PA_g}{\phi_k}\subset\gamma_k$ for $k=1,\dots,n.$
Then due to property 5 $\str\phi\gamma$ is nonempty.
In particular, $\strd {\phi_1}{\phi_2}$ is nonempty iff $\phi_1\sim\phi_2$, and $\str c\gamma$ is always nonempty.

\begin{figure}
\begin{center}
\includegraphics[width=.9\columnwidth]{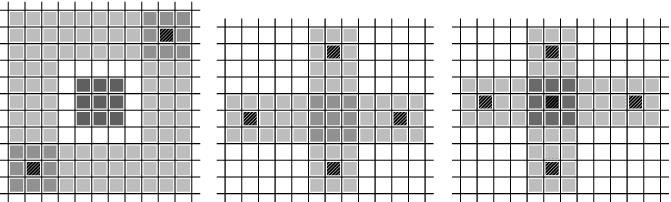}
\end{center}
\caption{
(Left) The geometry of \propref{inside} for $l=8$.
The striped sites are the common endpoints of two string operators $p_1$ and $p_2$, the first running trough the bottom and right side, the second through the top and right side.
The central shaded set of sites is $\bl 3 3 {l-2}{l-2}$.
(Center) The geometry of \propref{mutual} for $a_1=-5$, $a_2=4$, $b_1=-3$ and $b_2=4$.
The striped sites are the endpoints of two string operators $p_A$ and $p_B$.
(Right) The geometry of \propref{spin} for $a_1=-4$, $a_2=5$, $b_1=-5$ and $b_2=3$.
The striped sites are the endpoints of four string operators $p_i$, each with an endpoint in the central striped site.
}
\label{fig:inside}
\label{fig:mutual}
\label{fig:spin}
\end{figure}

The product of two strings with common endpoints gives a `closed' string operator.
Such a Pauli operator belongs to the centralizer of $\PA$.
We next show that when the two strings have nontrivial charge and enclose a block, the support of the generators of the resulting centralizer element must cover entirely the interior of the block.
The geometry is displayed in \fref{inside}.

\begin{prop}\label{prop:inside}
Let $\PB=\cen\PA\pauli$, $l\in\N^\ast$. Let $\phi\in c\in\charge\PA$ with $\suppr{\PA_g}\phi=\sset\Site$ and $c$ nontrivial, and $p_1,p_2\in\pauli$, $B\in\powersetfin {\PB_g}$ with
\be
p_1 \in\strd {\phi}{\trans l l (\phi)},\qquad 
p_2 \in\strd {\trans ll (\phi)}\phi,\qquad
p_1p_2\in\pro{B}.
\ee
Then $\bl 3 3 {l-2} {l-2}\subset  \supp {B}$.
\end{prop}

\begin{proof}
It is enough to consider $l\geq 6$. 
Let $\gamma:=\bl 3 3 {l-2} {l-2}$ and suppose that there is a site $\site i j \in \gamma-\supp {B}$.
Define the sets of sites
\be
\gamma_1:={\bl i j {i+1} \infty },\quad\gamma_2:={\bl {-\infty}{-\infty}{i}{\infty}},\quad \gamma_3:=\bl {i-1}{-1}{i+2}{3}.
\ee
Choose $b\in\pro {B\overlap {\gamma_1}}$ and set $p_3:=bp_1p_2$, so that 
$\supp b\subset\bl {i-1} {j}{i+2}\infty$ and $\gamma_1\cap \supp {p_3}=\emptyset.$
Choose $p,q,r\in\pauli$ with $p\propto p_1|_{\gamma_2}$, $q\propto p_3|_{\gamma_2}$.
Let $\Sigma_2:=\set{\bl a b {a+2}{b+2}}{\site a b\in\Sigma}$. Since $p_3\in\cen\PA\pauli$, $\suppr{\PA_g}{q}$ is contained in the set
\be
\bigcup \set {\gamma\in\Sigma_2} {\supp {p_3}\cap{\gamma_2}\cap\gamma\neq\emptyset, \supp {p_3}\cap(\Sigma-\gamma_2)\cap\gamma\neq\emptyset}\subset\gamma_3.
\ee
Let $\phi_1\in c$ be such that $\syndr{\PA}{p}=\phi_1\,\phi.$ Then $\suppr{\PA_g}{\phi_1}$ is contained in
\begin{align}
&\bigcup \set {\gamma\in\Sigma_2} {\supp {p}\cap{\gamma_2}\cap\gamma\neq\emptyset, \supp {p}\cap(\Sigma-\gamma_2)\cap\gamma\neq\emptyset}\cup\nl
&\quad\cup(\supp{\trans l l (\phi)}\cap\gamma_2)\subset\gamma_3.
\end{align}
But from $\thk(\gamma_3)\cap(\supp b\cup \supp{p_2})=\emptyset$ we get $\suppr{\PA_g}{q}=\suppr{\PA_g}{q}|_{\gamma_3}=\suppr{\PA_g}{p}|_{\gamma_3}$.
Thus $\syndr {\PA_g} {q} = \phi \in c$, a contradiction.
\end{proof}

\section{Topological charge}

This section explores how the commutation properties of string operators sometimes only depend on the charges and topology of the involved strings.

\subsection{String commutation rules}


We first consider the case of two crossing string operators, with charges possibly in two different $LPGs$.
When the two $LPG$-s are suitably related, whether the string operators commute depends only on the respective charges.
From a physical perspective, this corresponds to the `mutual statistics' of topological charges.
The geometry involved is displayed in \fref{mutual}.

\begin{prop}\label{prop:mutual}
Let $\PA,\PB$ be LPGs on a given lattice satisfying the properties in \propref{coarse_grain2} and 
\be\label{condition_mutual}
\cen\PB\pauli\subset\cen{\cen\PA\pauli}\pauli.
\ee
Let $p_A,p_B\in\pauli$ be such that, see \fref{mutual},
\begin{alignat}{2}
p_A \in\str {c_A}{\gamma_h},& \quad &&\gamma_h=\pth {\site {a_1}0,\site {a_2} 0},\nl
p_B \in\str{c_B}{\gamma_v},&\quad &&\gamma_v=\pth {\site 0{b_1},\site 0 {b_2}},
\end{alignat}
for some integers $a_1,b_1\leq -3$, $a_2,b_2\geq 3$, and $c_A\in\charge\PA$, $c_B\in\charge\PB$.
The quantity
\be
\mutual {c_A}{c_B}:=\comm{p_A}{p_B},
\ee
only depends on $c_A$ and $c_B$. Moreover, for any $c_A'\in\charge\PA$,
\be\label{properties_mutual}
\mutual {c_A}{c_B}=\mutual {c_B}{c_A},\qquad
\mutual {c_Ac_A'}{c_B}=\mutual {c_A}{c_B}\mutual {c_A'}{c_B}.
\ee
\end{prop}

\begin{proof}
Let $\bar p_A,\bar p_B, \bar a_1,\bar a_2, \bar b_1,\bar b_2$ be an alternative choice for the definition of $\mutual {c_A}{c_B}$, with the same properties as $p_A, p_B, a_1, a_2, b_1, b_2$.
Set $a:=a_2-\bar a_1$, $b:=b_2+\bar b_2+3$, and choose $q_A, q_B\in\pauli$ such that
\begin{align}
q_A &\in\str{\syndr{\PA}{p_A\bar p_A}}{\gamma_A},\,\,\, q_B \in\str{\syndr{\PB}{p_B\bar p_B}}{\gamma_B},\nl
\gamma_A&:=\pth{\site{a_1}{0},\site{a_1} {b},\site {\bar a_1+a} 0}\cup\bl{a_2}{0}{a_2+1}{1},\nl
\gamma_B&:=\pth{\site{0}{b_1},\site{a} {\bar b_1}}\cup\pth{\site{0}{b_2},\site{a} {\bar b_2}}.
\end{align}
We set $\tilde p_A:= \trans {a}0(\bar p_A)$, $\tilde p_B:= \trans {a}0(\bar p_B)$, so that
$p_A \tilde p_A q_A\in\cen\PA\pauli$, $p_B \tilde p_B q_B\in\cen\PB\pauli.$
Simple support considerations give
\begin{align}\label{example_comm}
\comm {p_A}{p_B}\comm {\bar p_A}{\bar p_B}=\comm {p_A}{p_B}\comm {\tilde p_A}{\tilde p_B} = \comm{p_A\tilde p_A q_A^1 q_A^2}{p_B \tilde p_B q_B^1 q_B^2}=1.
\end{align}
Regarding the first property in \eqref{properties_mutual}, first notice that condition \eqref{condition_mutual} is symmetric, so that $\mutual{c_B}{c_A}$ is defined.
Choose $\phi_A\in\charge\PA$, $\phi_B\in\charge\PB$ with support in $\sset \Site$.
Let $p_A, q_A, r_A, p_B, q_B, r_B \in\pauli$ be such that
\begin{align}
p_A&\in\strd {\trans {-6}0 (\phi_A)}{\trans {0}0 (\phi_A)},\quad
&q_A&\in\strd {\trans {0}{-6} (\phi_A)}{\trans {0}0 (\phi_A)},\nl
r_A&\in\strd {\trans {0}{-6} (\phi_A)}{\trans {-6}{0} (\phi_A)},\quad
&p_B&\in\strd {\trans {3}{-3} (\phi_B)}{\trans {-3}{-3} (\phi_B)},\nl
q_B&\in\strd {\trans {-3}{3} (\phi_B)}{\trans {-3}{-3} (\phi_B)},\quad
&r_B&\in\strd {\trans {-3}{3} (\phi_B)}{\trans {3}{-3} (\phi_B)}.
\end{align}
Again from support considerations
\begin{align}
\mutual {c_A}{c_B}\mutual {c_B}{c_A}=\comm {p_A}{p_B}\comm {q_A}{q_B}=\mutual{p_Aq_Ar_A}{p_Bq_Br_B}=1.
\end{align}
As for the second property in \eqref{properties_mutual}, just notice that for $c_1,c_2\in\charge\PA$, $\sigma, \sigma'\in\Sigma$,
\begin{align}
p_1\in&\strp {c_1} {\sigma}{\sigma'}\,\wedge\,p_2\in\strp {c_2} {\sigma}{\sigma'}\Longrightarrow\nl
&\Longrightarrow p_1p_2\in\strp {c_1c_2} {\sigma}{\sigma'}
\end{align}

\end{proof}

\begin{prop}\label{prop:exists_dual}
Let $\PA,\PB$ be as in \propref{mutual}.
Let $c_A\in \charge\PA-\sset 1$.
There exists $c_B \in \charge\PB$ with $\mutual {c_A}{c_B}=-1.$
\end{prop}
\begin{proof}
Choose $\PB_g$ independent, and $\phi\in c$, $h,v\in\pauli$, $B\in\powersetfin {\PB_g-\phase}$ with
\begin{alignat}{2}
h &\in\strd {\phi} {\trans 6 0 (\phi)},\qquad
&&\pro{B}\ni h\,\trans 0 6 (h)\, v\, \trans 6 0 (v), \nl
v &\in\strd {\phi}{\trans 06 (\phi)},\qquad
&&\suppr{\PA_g}{\phi}=\Site.
\end{alignat}
According to \propref{inside} we can choose $b\in B$ with $\supp {b}\subset \bl 3 3 4 4$.
Let $c_B$ be the charge of $b$ in $\PB_g$ and choose $L\in\N^\ast$ such that $\trans {-L} 0 (b)\not\in B$ and $p\in\str {\phi}{\trans 0 {-L}\,(\phi)}$, $\phi:=\mor{\PB_g} b$.
Then $\mutual {c_A}{c_B}=\comm p {v}=\comm p r=-1$ for any $r\in\pro B$ because $\supp p\cap\supp{h\trans 0 6(h)\trans 6 0 (v)}=\emptyset$.
\end{proof}

\begin{cor}\label{cor:duality}
$\charge\PA$ is dual to $\charge\PB$.
\end{cor}

Next we consider the case of string operators with common charge and a common endpoint.
The geometry involved is displayed in \fref{spin}.
From a physical perspective, the resulting invariant $\spin c$ corresponds to the `topological spin' of the topological charge $c$.
We say that $c$ is a boson if $\spin c=1$ and a fermion if $\spin c=-1$.

\begin{prop}\label{prop:spin}
Let $\PA$ be a LPG satisfying the properties in \propref{coarse_grain2} and
\be\label{condition_spin}
\cen\PA\pauli\subset\cen{\cen\PA\pauli}\pauli.
\ee
Let $c\in\charge \PA$, $P=\sset{p_1,p_2,p_3,p_4}\subset\pauli$ be such that for $k,l=1,2,3,4$
\begin{align}
p_k\in\str c {\gamma_k},\quad 
&\gamma_1 = \pth {\Site,\site{a_1}0},\quad
&\gamma_2 = \pth {\Site,\site{a_2}0},\nl 
p_kp_l \in\str c{\gamma_l\circ\gamma_k\inv},\quad
&\gamma_3 = \pth {\Site,\site{0}{b_1}},\quad 
&\gamma_4 = \pth {\Site,\site{0}{b_2}},
\end{align}
for some integers $a_1,b_1\leq -3$, $a_2,b_2\geq 3$.
Let $p,q,r\in P$ be all different.
The quantity
\be
\spin {c}:=\comm{p}{q}\comm{p}{r}\comm q r=\comm{pq}{pr},
\ee
only depends on $c$.
Moreover, for any $c_1,c_2\in\charge\PA$,
\begin{equation}\label{spin_mutual}
    \spin {c_1c_2}=\spin {c_1}\spin {c_2}\mutual{c_1}{c_2}.
\end{equation}
\end{prop}

\begin{proof}
There exist $\phi_k\in\morphfin{\PA}$ such that
\begin{align}
\supp{\phi_0}&\subset\bl 0 0 11,\quad&
\supp{\phi_1}&\subset\bl {a_1} 0 {a_1+1}1,\quad&
&\supp{\phi_2}\subset\bl {a_2} 0 {a_2+1}1,\nl 
\supp{\phi_3}&\subset\bl 0 {b_1} 1{b_1+1},\quad&
\supp{\phi_4}&\subset\bl 0 {b_2} 1{b_2+1},\quad&
&p_k \in\strd {\phi_0} {\phi_k},
\end{align}
where $k=1,2,3,4$.
Let $\bar \phi_k, \bar p_{k'}, \bar a_1,\bar a_2, \bar b_1,\bar b_2$ be an alternative choice for the definition of $\mutual {c_A}{c_B}$, with the same properties as $\bar \phi_k, p_{k'}, a_1, a_2, b_1,b_2$, $k=0,\dots,4$, $k'=1,\dots,4$.
We show that $\comm {p_1p_4}{p_2p_4}=\comm {\bar p_1\bar p_3}{\bar p_3\bar p_4}$, the other combinations are similar.
Set $b:=b_2-\bar b_1$ and choose $q_1, q_2,q_3\in\pauli$ such that
\begin{align}
q_1 &\in\strp{{\phi_1\trans 0b(\bar \phi_1)}}{\site{a_1}{0}}{\site{\bar a_1}{b}},\nl
q_2 &\in\strp{{\phi_2\trans 0b(\bar \phi_4)}}{\site 0{b+\bar a_2}}{\site{a_2}{0}},\nl
q_3 &\in\strp{{\phi_4\trans 0b(\bar \phi_3)}}{\site{0}{b_2}}{\site{0}{b+\bar b_1}}.
\end{align}
Then, from support considerations
\begin{align}
&\comm {p_1p_4}{p_2p_4}\comm {\trans 0 b (\bar p_1\bar p_3)}{\trans 0 b (\bar p_3\bar p_4)} =\\&\quad= \comm{p_1p_4 \trans 0b(\bar p_1\bar p_3) q_1q_3}{p_2p_4 \trans 0b(\bar p_3\bar p_4) q_2 q_3}=1.
\end{align}

We next prove \eqref{spin_mutual}.
Choose $\phi_1\in c_1,\phi_2\in c_2$ with support in $\sset \Site$.
Let $p, q, r, p', q', r' \in\pauli$ be such that
\begin{align}
p&\in\strd {\trans {0}0 (\phi_1)}{\trans {0}{-11} (\phi_1)},\qquad&
q&\in\strd {\trans {0}{0} (\phi_1)}{\trans {12}{0} (\phi_1)},\nl
r&\in\strd {\trans {0}{0} (\phi_1)}{\trans {0}{12} (\phi_1)},\qquad&
p'&\in\strd {\trans {3}{3} (\phi_2)}{\trans {3}{-8} (\phi_2)},\nl
q'&\in\strd {\trans {3}{3} (\phi_2)}{\trans {15}{3} (\phi_2)},\qquad&
r'&\in\strd {\trans {3}{3} (\phi_2)}{\trans {3}{15} (\phi_2)}.
\end{align}
Then $\comm {pp'qq'}{pp'rr'} = \comm {pq}{qr}\comm{p'q'}{q'r'}\comm{p'}q=\spin c\spin{c'}\mutual c{c'}.$
But, setting $\phi:=\coarse 4 {\phi_1\trans 3 3 (\phi_2)}$,  
\begin{align}
\coarse 4 {pp'}&\in\strd {\trans {0}0 (\phi)}{\trans {0}{-3} (\phi)},\qquad
\coarse 4 {qq'}&\in\strd {\trans {0}0 (\phi)}{\trans {3}{0} (\phi)},\nl
\coarse 4 {qq'}&\in\strd {\trans {0}0 (\phi)}{\trans {0}{3} (\phi)},
\end{align}
and, using the observation \eqref{coarse_charge} below, we get
\begin{align}
&\mutual {c}{c'}=\mutual {\coarse 4 c}{\coarse 4 c}
=\comm {\coarse 4 {pp'}\coarse 4 {qq'}}{\coarse 4 {pp'}\coarse 4 {rr'}} =\nl
&\quad=\comm{pp'qq'}{pp'rr'}= \comm {pq}{qr}\comm{p'q'}{q'r'}\comm{p'}q
=\spin c\spin{c'}\mutual c{c'}.
\end{align}
\end{proof}

Clearly $\mutual c 1=\spin 1 = 1$. Then \eqref{spin_mutual} gives $\mutual c c=1.$

\subsection{Local transformations}

$\mutualx$ and $\spinx$ are invariant under coarse graining, up to the natural isomorphism \eqref{coarse_natural},
\be\label{coarse_charge}
\mutualx = \mutualx\circ (\coarsex l\times\coarsex l),\qquad \spinx = \spinx\circ \coarsex l,
\ee
because $\coarse l {\str c \gamma}\subset \str{\coarse l c} {\coarse l \gamma}$ for any $c\in\charge\PA$, path $\gamma$ and $l\in\N^\ast$.
As for the behavior under LPG composition, consider the projections
$
\funct {p_i}{\charge{\PA_1\otimes\PA_2}}{\charge{\PA_i}},
$
constructed according to the isomorphism \eqref{composition_natural}.
When both sides are defined,
\be\label{tensor_charge}
\mutualx = (\mutualx\circ (p_1\times p_1))(\mutualx\circ (p_2\times p_2)), \qquad\spinx = (\mutualx\circ p_1)(\mutualx\circ p_2).
\ee
Finally, given an LPG isomorphism $F$ from $\PA$ to $\PB$ that can be regarded also as an LPG isomorphism $F'$ from $\PA'$ to $\PB'$, whenever both sides are defined
\be\label{map_charge}
\mutualx = \mutualx \circ (\charge F\times \charge {F'}),\qquad \spinx = \spinx\circ \charge F,
\ee
because, e.g., for any $c\in\charge\PA$, $L\in\Z$,
\begin{align}
\coarse {2r+1} {F[\str {c}{\gamma_1}]}\subset\str{\coarse {2r+1}{\charge F(c)}}{\gamma_2}
\end{align}
where $r=\range F-1$, $
\gamma_1:={\pth{\site{r}{r},\site{(2r+1)L+r}{r}}},
\gamma_2:={\pth{\Site,\site L0}}.
$

\section{Duality and canonical generators}

This section uncovers the relationship between $\charge \stab$ and $\charge\gauge$ and provides a classification of the topological charges that a TSSG may exhibit.

\subsection{Injection morphism}

Let $\stab$ be a TSSG with gauge group $\gauge$, which is a LPG according to \corref{centralizer}.
Using \thmref{independent}, \lemref{global_constraints}, \propref{constraint_duality} and \corref{duality} we learn that $\charge\stab$ and $\charge\gauge$ are finite and dual through $\mutualx$.
The functions
\be
\funct{\mutualx}{\charge\gauge\times\charge\gauge}{\pm 1},\qquad\funct{\spinx}{\charge\gauge}{\pm 1},\qquad
\funct{\mutualx}{\charge\gauge\times\charge\stab}{\pm 1},
\ee
are all defined. The charges in $\charge\gauge$ and $\charge\stab$ are naturally related. 
Since $\stab\subset\gauge$, there is an injection morphism $\funct \iota {\stab}{\gauge}$ that gives rise to the morphism
\functmap {\iota^\ast} {\morphfin {\gauge}}{\morphfin\stab}{\phi}{\phi\circ\iota.}
Since $\iota^\ast [\syndri\gauge\pauli]\subset\syndri\stab\pauli$, we get a natural morphism
\be\label{injection_morphism}
\funct {\charge\iota} {\charge {\gauge}}{\charge\stab}.
\ee

\begin{prop}\label{prop:iota_mutual}
$\mutual c d = \mutual c {\charge\iota(d)}$ for any $c,d\in\charge\gauge$.
\end{prop}

\begin{proof}
By coarse graining, we can always get $\stab_g$ and $\gauge_g$ such that \propref{coarse_grain1} holds for them and for every $s\in\stab_g$ there exists $G_s\subset \gauge_g$ with $s\in\pro {G_s}$ and $\supp {G_s}\subset\thk^1 (\supp s)$. 
But for any $c\in\charge\gauge$, $L\in\Z$,
\begin{align}
\coarse 5 {\str {c}{\pth{\site{2}{2},\site{5L+2}{2}}}}\subset\str{\coarse 5{\charge\iota(c)}}{\pth{\Site,\site L0}}.
\end{align}
\end{proof}

\begin{prop}\label{prop:iota_kernel}
For any $c\in\charge\gauge$, $\charge\iota(c)=1$ iff $\mutual c d =1$ for every $d\in\charge\gauge$.
\end{prop}

\begin{proof}
If $\charge\iota(c)=d\neq 1$ for some $d\in\charge\stab$, by duality there exists $e\in\charge\gauge$ such that $\mutual e d=-1=\mutual e c$. 
\end{proof}

\subsection{Canonical charge generators}\label{sec:canonical}

We want dual canonical sets of generators that reflect the properties of the injection morphism \eqref{injection_morphism}.
\begin{thm}\label{thm:charges}
Let $\stab$ be a TSSG with gauge group $\gauge$. For some $\alpha,Ê\beta\in \N_0$ and $\chi,f=\pm 1$ the groups $\charge\stab$ and $\charge\gauge$ admit independent set of generators
\begin{align}\label{canonical}
\charge\gauge=\sget {c_1,\dots,c_\alpha, d_1,\dots,d_\alpha,e_1,\dots,e_\beta}=\sget{c_i}_{i=1}^{2\alpha+\beta},\nl
\charge\stab=\sget{\tilde c_1,\dots,\tilde c_\alpha, \tilde d_1,\dots,\tilde d_\alpha, \tilde e_1,\dots,\tilde e_\beta}=\sget{\tilde c_i}_{i=1}^{2\alpha+\beta},
\end{align}
such that for $i=1,\dots, \alpha$, $k=1,\dots,\beta$ and $u,v=1,\dots,2\alpha+\beta$,
\begin{align}\label{canonical_properties}
\charge\iota(c_i)&=\tilde d_i,\quad&
\charge\iota(e_k)&=1,\quad&
&\spin {c_i}=\spin {d_i}=\chi^{\delta_{i1}},\nl
\charge\iota(d_i)&=\tilde c_i,\quad&
\spin {e_k}&=f^{\delta_{k1}},\quad&
&\mutual {c_u}{\tilde c_v}=1-2{\delta_{uv}}. 
\end{align}
Moreover, $\alpha$, $\beta$, $\chi$ and $f$ only depend on the TSSG, not the choice of generators.
\end{thm}

Before proceeding with the proof, notice that propositions \ref{prop:iota_mutual} and \ref{prop:iota_kernel} imply
\be\label{canonical_mutual_gauge}
  \mutual {c_i}{c_j}=\mutual {d_i}{d_j}=\mutual {c_u}{e_k}=1,\quad \mutual{c_i}{d_j}=1-2{\delta_{ij}}.
\ee

\begin{proof}
We find the generators in three steps.

\noindent {\bf 1)} 
Let $K$ be the kernel of $\charge\iota$ in $\charge\gauge$. 
$K\simeq \Z_2^\beta$ for some $\beta\in\N_0$ and it admits a set of independent generators $K=\sget{e_1,\dots,e_\beta}$.
If $\theta(e_k)=1$ for every $k=1,\dots, n$, we do nothing and set $f=1$.
In other case, suppose w.l.o.g. that $\theta(e_1)=-1$ and take new generators $\bar e_k$ with $\bar e_k=e_k$ if $\spin {e_k}=1$ or $k=1$, $\bar e_k=e_k e_\beta$ otherwise, remove the bars and set $f=-1$.
Since $f=1$ iff all the elements of $K$ are bosons, $f$ does not depend on the choice of generators.

\noindent {\bf 2)} 
$\charge\gauge\simeq \Z_2^{\beta+n}$ for some $n\in \N_0$ and we choose a set of independent generators $\gauge=\sget{e_1,\dots,e_\beta, g_1,\dots,g_n}$.
According to \propref{iota_kernel}, for each $g_i$ there exists $g_j$ such that $\mutual {g_i}{g_j}=-1$, $i,j=1,\dots, n$.
In a similar way that a canonical basis of the Pauli group is obtained, we can obtain from the $g_i$ a set of generators of $\charge\gauge$ as in \eqref{canonical} and with the mutual statistics of \eqref{canonical_mutual_gauge}.
If $\spin {c_i}\spin {d_i}=1$ does not hold, we can always find suitable generators. 
E.g., if $\spin {c_i}=-\spin {d_i}=1$ we set $\bar c_i:=c_i, \bar d_i:=c_id_i$ and remove the bars.
Similarly, suppose that there exist $1\leq i<j\leq n$ with $\spin {c_i}=\spin {c_j}=-1$.
Then we can set $\bar c_i:=c_i c_j$, $\bar d_i:=d_i c_j$, $\bar c_j:=c_i d_i d_j$ and $\bar d_j:=c_i d_i c_j d_j$ so that $\spin {\bar c_i}=\spin {\bar d_i}=\spin {\bar c_j}=\spin {\bar d_j}=1$, and remove the bars. 
$\chi$ only depends on the TSG because the total number of bosos in $\charge\gauge$ is $2^{\alpha+\beta-1}(2^{\alpha+1}+\chi+\chi f).$

\noindent {\bf 3)}
We define $\tilde c_i$ and $\tilde d_i$ according to \eqref{canonical_properties}.
The existence of the $e_k$ is a consequence of duality, which also guarantees the independence of all these generators.
\end{proof}

\begin{defn}
The characteristic of a TSSG is given by the numbers $\alpha,\beta\in\N_0$ and $\chi,f=\pm 1$ in \thmref{charges}, with $\chi=1$ if $\alpha=0$ and $f=1$ if $\beta=0$.
\end{defn}

The composition of two TSSGs with characteristics $\alpha^k,\beta^k,\chi^k,f^k$, $k=1,2$, has characteristic
\begin{align}
\alpha=\alpha^1+\alpha^2,\quad
\beta=\beta^1+\beta^2,\quad
\chi=\chi^1\chi^2,\quad
2(1+f)=(1+f^1)(1+f^2).
\end{align}
The examples in table \ref{tab:characteristic} show that TSSGs of arbitrary characteristic exist.

\begin{table}[h]
\centering
\begin{tabular}{|l|c|c|c|c|}
 \hline
  Code & $\alpha$ & $\beta$ & $\chi$ & $f$ \\
  \hline
  Empty code, trivial code, trivial subsystem code & 0 & 0 & 1 & 1 \\
  Toric code & 1 & 0 & 1 & 1 \\
  Subsystem toric code & 0 & 1 & 1 & 1 \\
  Fermionic subsystem toric code, honeycomb code & 0 & 1 & 1 & -1 \\
  Topological subsystem color code & 1 & 0 & -1 & 1 \\
   \hline
\end{tabular}
\caption{
The characteristic of several codes.
}
\label{tab:characteristic}
\end{table}

\subsection{Local equivalence}

Two TSSGs $\stab_1,\stab_2$ have the same characteristic iff there exists a group isomorphism
$\funct{\lambda} {\charge{\gauge_1}}{\charge{\gauge_2}}$ such that 
\be
\mutualx\circ (\lambda\times\lambda)=\mutualx, \qquad\spinx\circ \lambda=\spinx.
\ee
In view of \eqref{map_charge} a LPG isomorphism $F$ from $\stab_1$ to $\stab_2$ induces such an isomorphism $\lambda=\charge {F}$, with $F$ regarded as an isomorphism from $\gauge_1$ to $\gauge_2$.
The same is true for coarse graining, see \eqref{coarse_charge}.
Finally, the composition of a TSSG with a trivial code also gives a charge isomorphism in view of \eqref{tensor_charge}.
In summary:

\begin{prop}
Locally equivalent TSSGs have the same characteristic.
\end{prop}

\subsection{Chirality}

A TSG $\stab$ can be regarded as a TSSG with $\beta=0$, so that its characteristic is given by $\alpha$ and $\chi$.
But all known TSGs have $\chi=1$ (e.g., for toric codes $\alpha=1$ and for color codes $\alpha=2$).
Indeed, from the condensed matter perspective the Hamiltonian model related to a TSG is chiral if $\chi=-1$ \cite{bombin:2011:universal}, and this is thought to be incompatible with the fact that the stabilizer generators commute with each other.
Therefore, unlike in TSSGs not all characteristics may admit a realization.
On the other hand, if we consider only non-chiral TSGs, those with $\chi=1$, there exists TSGs with arbitrary values of $\alpha$.
We will generally refer to TSSGs with $\chi = -1$ as chiral.

\section{Structure}\label{sec:classification}

This section shows how any TSSG can be put in a standard form by means of a framework of string operators.

\subsection{More coarse graining}

We will find useful the following notation for the indices of canonical charges as given in \thmref{charges}:
\begin{alignat*}{2}
\indc&:=\sset{1,\dots,\alpha},\quad
&\inde&:=\sset{2\alpha+1,\dots,2\alpha+\beta},\nl
\indd&:=\sset{\alpha+1,\dots,2\alpha},\quad
&\indall&:=\indc\cup \indd\cup \inde,
\end{alignat*} 
\be\label{notation_k}
k^\ast:= \left\{
     \begin{array}{ll}
        k+\alpha&\,\mathrm{ if }\,\, k\in \indc,\\
        k-\alpha&\,\mathrm{ if }\,\, k\in \indd,\\
        k&\,\mathrm{ if }\,\, k\in \inde,\\
     \end{array}
 \right.
\ee

Before proceeding with the main result we need to gain an additional property by coarse graining.
Namely, for $k\in \indc\cup \indd$ and given canonical charge generators, not only we should be able to find, in any given site, separately a stabilizer generator $s_k$ with charge $\tilde c_k$ and a gauge generator $g_{k^\ast}$ with charge $c_{k^\ast}$.
Rather, both should correspond to the same $\phi$ up to the morphism $\iota^\ast$.
\begin{prop}\label{prop:charge_basis_site_TSSG}
Any TSSG can be coarse grained to a TSSG $\stab$, with gauge group $\gauge$ and LTI set of independent generators $\stab_g$,  $\gauge_g$, satisfying the properties in \propref{coarse_grain2} and also the following, for a choice of canonical charge generators.
There exist $s_k\in\stab_g$ and $g_k\in\gauge_g-\phase$, $k\in \indall$, with 
\be
\chgr{\stab_g} {s_k}= \tilde c_k,\qquad \chgr{\gauge_g} {g_k}= c_k,\qquad\supp{s_k}= \supp{g_k}=\sset\Site,
\ee
and, for $k\in \indc\cup \indd$,
\be\label{charge_accord}
{\sset{s_k}}=\negr{\stab_g} {\iota^\ast (\mor{\gauge_g} {\sset{g_{k^\ast}}})}.
\ee
\end{prop}

\begin{proof}
Coarse graining and choosing $\gauge'_g,\stab_g'$ according to \propref{coarse_grain2} for certain $c_k',\tilde c_k'$, there exist $s_k', g_k'$ that satisfy all conditions except \eqref{charge_accord}. 
Let $l\in\N^\ast$ be the minimal integer such that for any $k\in \indc\cup \indd$
\be
\suppr{\stab_g'}{\iota^\ast(\mor{\gauge_g'} {\sset{g_{k}'}})}\subset\thk^l(\sset\Site).
\ee
In order that all properties except \eqref{charge_accord} are met, let 
\begin{align}
\gauge_g&:=\coarse {2l+1} {\gauge_g'},\quad&
c_k&:=\coarse {2l+1} {c_k'},\quad&
g_k&:=\coarse {2l+1} {\trans l l (g_k')}, \nl
\stab_g''&:=\coarse {2l+1} {\stab_g'},\quad&
\tilde c_k&:=\coarse {2l+1} {\tilde c_k'},\quad&
s_k''&:=\coarse {2l+1} {\trans l l (s_k')}.
\end{align}
For any $k\in \indc\cup \indd$ we have
$\suppr{\stab_g''}{\iota^\ast (\mor {\gauge_g'}{g_k}}=\sset\Site.$
For $k\in \inde$ we set $s_k:=s_k''$.
Let us show that there exist a list of LTI sets of independent generators $\stab_g''=\stab_g^0,\stab_g^1,\dots, \stab_g^{2\alpha}=\stab_g$
such that for any $k\in K$ and $\stab_g^t$ with $k>2\alpha-t$ there exist a suitable $s_k$.
We construct $\stab_g^{t+1}$ from $\stab_g^{t}$ as follows.
Set $u:=2\alpha-t$, $\phi:=\iota^\ast(\mor{\gauge_g}{g_{u^\ast}})$ and choose $s_u\in \negr {\stab_g^{t}}\phi$ such that $s_u\neq s_{k}$ for any $k\in K$ with $k>u$.
This is always possible because the charge $\tilde c_{u}$ is independent of the charges $\tilde c_{u+1},\dots,\tilde c_{2\alpha+\beta}$.
Label the elements of $\negr {\stab_g^{t}}\phi-\sset{s_u}$ as $\sset{\hat s_i}_{i=1}^r$ and apply \propref{charge_change} repeatedly to perform the substitutions 
$\hat s_i\rightarrow \hat s_u\hat s_i$, $i=1,\dots, r$. 
Letting $\stab_g^{t+1}$ be the resulting set of generators, we get as needed 
\be
\mor{\stab_g^{t+1}} {\sset{s_{u}}}=\iota^\ast (\mor{\gauge_g} {\sset{g_{u^\ast}}})\in \tilde c_{u}.
\ee
Notice that the transformations $\hat s_i\rightarrow \hat s_0\hat s_i$ may affect the $s_k$, $k>u$, that had the required properties in $\stab_g^{t}$, but does not affect these properties.
\end{proof}

\subsection{Homological structure}

It turns out to be useful to consider an extension of $\mutualx$ and $\spinx$ to all charges $\charge \stab$.
We fix part of this extension setting $\mutual {\tilde c_k} \cdot=\mutual {c_{k^\ast}}\cdot$ for $k\in K$ and $\spin{\tilde c_k}=\spin{c_{k^\ast}}$ for $k\in \indc\cup \indd$.
This still leaves room for choosing $\spin{\tilde c_k}=\pm 1$ for $k\in \inde$ in any way that we please.

With such an extension in hand, we now give a `homological' construction from which the generators of $\gauge$ and $\stab$ with nontrivial charge can be recovered in a convenient way.
The idea is as follows.
For each $k\in \indc\cup \inde$ we visualize an infinite square lattice $\graph$ and its dual $\dgraph$, and attach dual charges to them.
For each edge and dual edge there is an operator in $\pauli$, and the commutation relations of these operators are those that we would expect were they string operators with the corresponding charge.
In particular, for $k\in \indc$ direct edges carry both charge $c_k$ and $\tilde d_k$, and dual edges both charge $d_k$ and $\tilde c_k$.
As for $k\in \inde$, direct edges carry charge $e_l$ and dual edges charge $\tilde e_l$, where $k=l+2\alpha$.
Edge operators from different pairs of dual lattices commute.
Closed strings in a given lattice give rise to elements of $\gauge$ or $\stab$,
and it is possible to find $\stab_g$ and $\gauge_g$ such that all charged elements are `face operators'.

\begin{thm}\label{thm:lattice}
Let $\stab$ be a TSSG with gauge group $\gauge$, characteristic $\alpha,\beta,\chi,f$ and canonical generators \eqref{canonical}.
There exist $L\in\N^\ast$, 
independent generator sets $\stab_g$, $\gauge_g$, and mappings 
\be\label{lattice_edges}
\funct {\edgeopx k}{\edges\cup\dedges}{\pauli},\qquad k\in \indc\cup \inde
\ee
that satisfy the following properties.

\begin{enumerate}
\item 
The mappings are translationally invariant,
\be\label{lattice_TI}
\trans {iL}{jL} \circ \edgeopx k =\edgeopx k \circ \trans i j.
\ee

\item 
For any $d,e\in\edges$ and $k,l\in \indc\cup \inde$,
\begin{align}\label{lattice_edges_comm}
\comm{\edgeop k {d}}{\edgeop l {e}}&=
 \left\{
     \begin{array}{ll}
       \spin{c_k} &\text{if } l= k\text{ and }|\partial d\cap\partial e|= 1,\\
       1 & \text{otherwise,}
     \end{array}
 \right.\nl
\comm{\edgeop k {d^\ast}}{\edgeop l {e^\ast}}&=
 \left\{
     \begin{array}{ll}
       \spin{\tilde c_k} &\text{if } l=k\text{ and }|\partial d^\ast\cap\partial e^\ast|= 1,\\
       1 & \text{otherwise,}
     \end{array}
 \right.\nl
\comm{\edgeop k {d}}{\edgeop l {e^\ast}}&=
 \left\{
     \begin{array}{ll}
       -1\quad &\text{if } l= k\text{ and } d= e,\\
       1 & \text{otherwise.}
     \end{array}
 \right.
\end{align}

\item
Let $\stab_g(c)\subset\stab_g$ contain all elements of charge $c$, and similarly for $\gauge_g(c)$.
For $k\in \indc$, $l\in \inde$,
\begin{alignat}{2}\label{lattice_sg_generators_k}
\gauge_g(c_k)&=\stab_g(\tilde d_k) \propto \bigcup_{f^\ast\in\dfaces} \pro {\edgeopi k {\partial f^\ast}},\quad
&\stab_g(\tilde c_l)  &\propto \bigcup_{f\in\faces} \pro {\edgeopi l {\partial f}},\nl
\gauge_g(d_k)&=\stab_g(\tilde c_k)  \propto \bigcup_{f\in\faces} \pro {\edgeopi k {\partial f}},\quad
&\gauge_g(c_l) &\propto \bigcup_{f^\ast\in\dfaces} \pro {\edgeopi l {\partial f^\ast}}.
\end{alignat}

\item
Let $\pauli_k := \sget {\edgeop k e}_{e\in\edges}$ and $\pauli_k^\ast := \sget {\edgeop k e}_{e\in\dedges}$. Then
\begin{alignat}{2}\label{lattice_sg_generators_1}
\gauge_g(1) &\subset \cen{\pauli_\graph'}\pauli\qquad
&\pauli_\graph'&:=\prod_{k\in \indc} \pauli_k\pauli_k^\ast\prod_{l\in \indd} \pauli_l,\nl
\stab_g(1) &\subset \cen{\pauli_\graph}\pauli,\qquad 
&\pauli_\graph&:=\prod_{k\in \indc\cup \indd} \pauli_k\pauli_k^\ast.
\end{alignat}

\end{enumerate}
\end{thm}

\begin{proof}
We proceed in three steps.
\noindent {\bf 0)} 
We assume that $\gauge=\sget{\gauge_g'}$ and $\stab=\sget{\stab_g'}$, with $\gauge_g'$ and $\stab_g'$ LTI sets of independent generators that have been coarse grained and chosen according to \propref{charge_basis_site_TSSG}.
We can assume such a coarse graining w.l.o.g because it can be absorbed in the value of $L$.

\noindent {\bf 1)} 
We first give the construction of $\edgeopx k$, and then explain it.
Let $L=3(2\alpha+5\beta)$ and set 
$\phi_k:=\mor{\gauge_g'} {\sset{g_{k}}},$ $\tilde\phi_k:=\mor{\stab_g'} {\sset{s_{k}}}$ for $k\in \indall$.
Choose for $k\in \indc\cup \indd$, $l\in\inde$,
\begin{alignat}{3}\label{edge_raw}
H_k&\in\strd {\phi_k}{\trans L0(\phi_k)},\quad &V_k&\in\strd{\phi_k}{\trans 0L(\phi_k)},\quad &h_l&\in\strd {\phi_l}{\trans 60(\phi_l)},\nl
H_l&\in\strd {\phi_l}{\trans L0(\phi_l)},\quad &V_l&\in\strd{\phi_l}{\trans 0L(\phi_l)},\quad &v_l&\in\strd{\phi_l}{\trans 06(\phi_l)},\nl
\tilde H_l&\in\strd{\tilde\phi^l}{ \trans L0(\tilde\phi_l)},\quad &\tilde V_l&\in\strd{\tilde\phi_k}{\trans 0L(\tilde\phi_l)},\quad&p_l&\in\strd{\iota^\ast(\phi_l)}1.
\end{alignat}
Define for $k\in \indc\cup \indd$, $l,l'\in\inde$,
\begin{alignat}{2}\label{binary_corrections}
2\,\theta_0^k&:=1-\spin{c_k}\,\comm {H_k}{\trans L 0 (H_k)},\qquad &2\,\theta_{i,j}^l&:=1-\spin{\tilde c_l}\,\comm {D_i^l}{D_j^l},\nl
2\,\theta_1^k&:=1-(1-2\theta_0^k)\,\comm {H_k}{V_k},\qquad &2\,\tau_{l,l'}&:=1-\comm {H_l}{\trans 3 3 (V_{l'})},\nl
2\,\theta_2^k&:=1-(1-2\theta_1^k)\,\comm {V_k}{\trans 0 L (V_k)},
\end{alignat}
where $i,j\in\sset {\mathrm {E(ast), \mathrm N(orth),\mathrm W(est),\mathrm S(outh)}},$ and
\be
D_{\mathrm E}^l:=\tilde H_l,\quad D_{\mathrm N}^l:= \tilde V_l,\quad D_{\mathrm W}^l:=\trans {-L} 0 (\tilde H_l),\quad D_{\mathrm S}^l:=\trans 0{-L}(\tilde V_l).
\ee
Let $\Delta:=-18\alpha+3\beta$. Define $\edgeopx k$ according to \eqref{lattice_TI} setting for $k\in \indc$, $l\in \inde$
\begin{alignat}{2}\label{define_edges}
\edgeop k \hedge &:=\Upsilon_{3k,3k}^{L,0}(s_k, \theta_0^k, 0, H_k),\qquad &
\dedgeop {k} \hedge&:=\Upsilon_{3k^\ast,3k^\ast-L}^{0,L}(s_{k^\ast},\theta_1^{k^\ast},\theta_2^{k^\ast},V_{k^\ast}),\nl
\edgeop k \vedge&:=\Upsilon_{3k,3k}^{0,L}(s_k, \theta_1^k,\theta_2^k, V_k),\qquad &
\dedgeop {k} \vedge &:=\Upsilon_{3k^\ast-L,3k^\ast}^{L,0}(s_{k^\ast},{\theta_0(k^\ast)},0,H_{k^\ast}),\nl
\edgeop {l} \hedge &:=\Upsilon_{3l,3l}^{L,0}(p_l,1,1,H_l),\qquad&
\dedgeop {l} \hedge &:=\trans{12l+\Delta}{12l+\Delta-L}(\tilde v_l\tilde V_l),\nl
\edgeop {l} \vedge &:=\Upsilon_{3l,3l}^{0,L}(p_l,1,1,V_l),\qquad&
\dedgeop {l} \vedge &:=\trans{12l+\Delta-L}{12l+\Delta}(\tilde h_l\tilde H_l),
\end{alignat}
\begin{align}
\tilde h_l:= \hat h_l^{\theta^l_{\mathrm{E,S}}}\,\trans L0(\hat h_l^{\theta^l_{\mathrm{W,S}}}\hat v_l^{\theta^l_{\mathrm{W,E}}})\prod_{d=1}^{2\alpha+\beta-l} \trans {3d-L,3d}(\hat v_l),&\qquad
\hat h_l := \Upsilon_{-3,6}^{6,0}(p_l,1,1,h_l),\nl
\tilde v_l:= \hat h_l^{\theta^l_{\mathrm{N,S}}}\hat v_l^{\theta^l_{\mathrm{N,E}}}\,\trans {-L}0(\hat v_l^{\theta^l_{\mathrm{N,W}}})\prod_{d=1}^{2\alpha+\beta-l} \trans {3d,3d}(\hat h_l),&\qquad
\hat v_l := \Upsilon_{6,-3}^{0,6}(p_l,1,1,v_l),\nl
\Upsilon_{x,y}^{u,v}(c,t,t',W):=\trans x y\left(c^t \trans u v (c^{t'})W\right).&
\end{align}
The main idea behind this constructions is to build the edge operators from suitable string operators, as those in \eqref{edge_raw}, placing them in lattices that are shifted from each other to guarantee suitable commutation relations.
As for the details, firstly we seek $\edgeop {l} e\in\gauge$ for $e\in\edges$, which is the reason for the introduction of $p_l$.
Secondly, we want the equations in \eqref{lattice_edges_comm} to hold, so we introduce the binary values in \eqref{binary_corrections} that tell us what adjustments are needed.
Regarding corrections involving the topological spin, marked $\theta$, different values of $k$ require different adjustments.
For $k\in \indc$ only three numbers describe the adjustments because of the constraints imposed by \propref{spin}, and the adjustment can be made using $s_k$.
For $k\in\inde$ and dual edges there are six adjustments to make, and they can be made with the help of suitably placed `tiny' string operators $\hat h_l, \hat v_l\in\gauge$ that only share support with one of the `edge' string operators in the construction.
The reason for the factor 12 in \eqref{define_edges} ---rather than 3, the width of a string---, is indeed the need to make room for these tiny strings.
For $k\in\inde$ and direct edges no adjustments are needed.
E.g., if we take $h={\edgeop l \hedge}, \quad v={\edgeop l \vedge}, \quad e=\trans {-L}0 (h)$ and $s:=\pro {\edgeopi l \face} \in\phase \stab$, then $\comm e h \comm e v = \comm e {hv}=\comm e s = 1$ and thus $\comm e h=\comm e v$.
Similarly one can get $\comm h v=\comm e v$ and thus $\comm e h=\spin {c_l}$ using \propref{spin}.
Finally, corrections regarding mutual statistics, marked $\tau$, are only needed for $k\in\inde$ and dual edges, and we perform them using $\hat h_l, \hat v_l\in\gauge$ again.

\noindent {\bf 2)}
We next construct $\stab_g$, $\gauge_g$.
Let for $i,j\in\Z$, $k\in K$,
\begin{align}
s^k_{i,j}&:=
 \left\{
     \begin{array}{ll}
       \trans{iL+3{k^\ast}}{jL+3{k^\ast}} (s_{k}) &\text{if } k\in \indc\cup \indd,\\
       \trans{iL+12k+\Delta}{jL+12K+\Delta} (s_{k}) & \text{otherwise,}
     \end{array}
 \right.\nl
g^k_{i,j}&:=\trans{iL+3{k}}{jL+3{k}} (g_{k}),\qquad
p^k_{i,j}:=\trans{iL+3{k}}{jL+3{k}} (p_{k}),
\end{align}
Choose $S_k\in\stab$ and $G_k\in\gauge$ for $k\in \indall$ such that
\begin{align}
S_k&\propto \left\{
     \begin{array}{ll}
       \pro{\edgeopi k {\face}} &\,\mathrm{ if }\,\, k\in \indc,\\
       \pro{\edgeopi {k^\ast} \dface} &\,\mathrm{ if }\,\, k\in \indd,\\
       \pro{\edgeopi {k} \face} &\,\mathrm{ if }\,\, k\in \inde.\\
     \end{array}
 \right.
\qquad
G_k&\propto \left\{
     \begin{array}{ll}
       \pro{\edgeopi k \dface} &\,\mathrm{ if }\,\, k\in \indc,\\
       \pro{\edgeopi {k^\ast} \face} &\,\mathrm{ if }\,\, k\in \indd,\\
       \pro{\edgeopi {k} {\dface}} &\,\mathrm{ if }\,\, k\in \inde.\\
     \end{array}
 \right.
\end{align}
As we show below, for $k,k'\in K$, $i,j\in \Z$,  
\be\label{assume}
s^k_{i,j}\in \hat S^{k'}\iff g^k_{i,j}\in \hat G^{k'}\iff i=j=k-k'=0,
\ee
where $\hat S^k\in\powersetfin{\stab_g'}$, $\hat G^k\in\powersetfin{\gauge_g'}$ are such that $S^k\in\pro{\hat S^k}$, $G^k\in\pro{\hat G^k}$.
Applying repeatedly \propref{charge_change} we can define new independent generators
\begin{align}
\stab_g:=\stab_g'\cup\bigcup_{k\in K}\strans L (\sset{S_k})-\bigcup_{k\in K}\strans L (\sset{s^k_{0,0}}),\nl
\gauge_g'':=\gauge_g'\cup\bigcup_{k\in K}\strans L (\sset{G_k})-\bigcup_{k\in K}\strans L (\sset{g^k_{0,0}}).
\end{align}
such that $\chgr{\stab_g}{S_k}=\tilde c_k$ and $\chgr{\gauge_g}{G_k''}=c_k$. 
The rest of generators have trivial charge because for any $k\in K$
\be
\strans L (\sset{S_k})\in \cnstr {\stab_g},\quad \strans L (\sset{G_k})\in \cnstr {\gauge_g''}.
\ee

For $l\in \inde$, let $P_l=\negr{\gauge_g''}{p^l_{0,0}}$. Clearly ---consider the support--- $G_k\not\in \strans L {P_l}$ for any $k\in \indall$.
Then applying repeatedly \propref{charge_change} we can define
\be
\gauge_g:= {\gauge_g^0}''\cup\bigcup_{l\in \inde}\bigcup_{p\in P_l}\strans L (pG_l)-\bigcup_{l\in \inde}\bigcup_{p\in P_l}\strans L (p).
\ee
This last change does not involve charged generators and thus completes the construction, which clearly satisfies the desired properties.

To prove \eqref{assume} we define the semi-infinite strings
\be
E_{i,j}:=\bigcup_{n\in\N_0}\trans{i+n}{j} (\hedge),\quad
E_{i,j}^\ast:=\bigcup_{n\in\N^\ast}\trans{i+n}{j} (\vedge^\ast).
\ee
Then for any $k,k'\in \indc$, $l,l'\in \inde$, by construction
\begin{align}
\negr{\gauge_g'}{\edgeopi {k} {E_{i,j}}}=\negr{\stab_g'}{\edgeopi {k} {E_{i,j}}}&=\sset{g^{k}_{i,j}}=\sset{s^{k^\ast}_{i,j}},&
&\negr{\gauge_g'}{\edgeopi {l} {E_{i,j}}\,p^l_{i,j}}=\sset{g^{l}_{i,j}},\nl
\negr{\gauge_g'}{\edgeopi {k} {E_{i,j}^\ast}}=\negr{\stab_g'}{\edgeopi {k} {E_{i,j}^\ast}}&=\sset{g^{k^\ast}_{i,j}}=\sset{s^{k}_{i,j}},&
&\negr{\stab_g'}{\edgeopi {l} {E_{i,j}^\ast}}=\sset{s^{l}_{i,j}},
\end{align}
and due to \eqref{lattice_edges_comm}
\begin{align}
&\comm {\edgeopi {k'} {E_{i,j}}}{G_k}=
\comm {\edgeopi {k'} {E_{i,j}}}{S_{k^\ast}}=
\comm {\edgeopi {k'} {E_{i,j}^\ast}}{G^{k^\ast}}
=\comm {\edgeopi {k'} {E_{i,j}^\ast}}{S^{k}}=1-2\delta_{i,0}\delta_{j,0}\delta_{k,k'},\nl
&\comm {\edgeopi {l'} {E_{i,j}}p^l_{i,j}}{G^l}=
\comm {\edgeopi {l'} {E_{i,j}^\ast}}{S^l}=1-2\delta_{i,0}\delta_{j,0}\delta_{l,l'},
\end{align}
which together give the desired result by inspection.
\end{proof}

According to \eqref{lattice_sg_generators_k} we can label non-trivially charged elements of $\stab_g$ and  $\gauge_g$ with a face or dual face, so that $\stab_g(\tilde c_k)=\sset{s^k_f}_f$ and $\gauge_g(c_k)=\sset{g^k_f}_f$,
where for each $k$ the index $f$ takes values either in $\faces$ or $\dfaces$ and $g^k_f=s^{k^\ast}_f$ for $k\in\indc$.
Then for $E\subset \edges$, $k\in\indc\cup\inde$,
\begin{align}
\negr{\mathcal F_g}{\edgeopi k {E}}&=\set{s^k_{f^\ast}}{f^\ast\in\partial E},\quad
\negr{\mathcal F_g}{\edgeopi k {E^\ast}}&=\set{g^k_f}{f\in\partial E^\ast},
\end{align}
where $\mathcal F_g:=\stab_g\cup(\gauge_g-\gauge_g(1))$. 
Define
\begin{alignat}{2}
\hedgeop {k}{e} &:= \edgeop k {e^\ast} (g^k_{f} g^k_{h})^{b(c_k)},\qquad&\partial e^\ast &= \sset{f,h},\nl
\hedgeop {k}{e^\ast} &:= \edgeop k {e} (s^k_{f^\ast} s^k_{h^\ast})^{b(\tilde c_k)},\qquad&\partial e &= \sset{f^\ast,h^\ast},
\end{alignat}
where $b(\cdot)=(1-\spin{\cdot})/2$, $k\in\indc\cup\inde$ and $e\in\edges$. 
Then for $d,e\in\edges\cup\dedges$ and $k,l\in\indc\cup\inde$,
\be\label{dual_edges}
\comm{\hedgeop {k}{d}}{\edgeop l e} = 1-2\delta_{k,l}\delta_{d,e},
\ee
which implies that the edge operators $\edgeop k e$ form a LTI set $\mathcal E$ of independent generators of $\pauli_\graph$ with $\cnstr {\mathcal E}=\emptyset$.
Dual edge operators $\hedgeop k e\in\pauli_\graph$ satisfy the same properties as the edge operators $\edgeop k {e^\ast}$.
More importantly, \eqref{dual_edges} implies that any element of $\pauli$ can be uniquely decomposed ---up to a phase--- as $pq$ with $p\in\pauli_\graph$ and $q\in\cen{\pauli_\graph}\pauli$.
This paves the way for the following complement to \thmref{lattice}, which isolates the trivial part of the stabilizer.

\begin{prop}\label{prop:stab_injection}
There exists a LPG morphism $\funct F{\paulitrivial^{\otimes n}}{\coarse L\pauli}$ 
\be
\morphism F {\stabtrivial^{\otimes n}}{\coarse L{\sget{\stab_g(1)}}},
\ee
with period 1 and $F[\paulitrivial^{\otimes n}]\subset\coarse L {\cen {\pauli_\graph} {\pauli}}$, for some $n\in \N_0$.
\end{prop}

\begin{proof}
Due to the coarse graining, we can assume w.l.o.g. that $L=1$. 
Due to the above considerations, for each $s\in\stab_g(1)$ we can choose $\overline s\in \cen {\pauli_\graph}\pauli$ such that $\negr{\stab_g}{\syndr\stab{\overline s}}=\sset{s}$.
Moreover, the choice can be done in a translationally invariant way so that $\overline {\trans {i}{j} (s)}=\trans{i}{j} (\overline s)$.
Consider a translationally invariant ordering of $\stab_g(1)$, so that for $r,s\in\stab_g(1)$ we have $r<s$ iff $\trans {i} {j} (r)<\trans {i} {j} (s)$.
For $s\in\stab_g^0$, introduce the \emph{finite} sets $\mu(s):=\set{r\in\stab_g(1)}{r<s, \comm {\overline r}{\overline s}=-1}$, and let $\hat s:=\bar s\prod_{r\in\mu(s)} r$, which clearly is a translationally invariant definition.
Then for any $r,s\in\stab_g(1)$ we have
\begin{equation}
\comm {r} {s}=1,\quad
\comm {\hat r} {\hat s}=1,\quad\comm {r} {\hat s} =1-2 \delta_{r,s}.
\end{equation}
Given $s\in\stab_g(1)$, let  $\sigma_s\in\Sigma$ be the minimal element of $\supp s$ in lexicographical ordering, which is translationally invariant.
Let us label the elements of $\stab_1$ with $\sigma_s=\Site$ as $s^m$, with $m=1,\dots,n$ for some $n\in\N_0$.
We can then define $F$ setting $F(Z_{i,j}^m)=s^m$ and $F(X_{i,j}^m)=\psi^m \hat s^m$, where $\psi^m=1,i$ as needed and $X_{i,j}^m,Z_{i,j}^m$ denote $X_{i,j},Z_{i,j}$ on the $m$-th copy of $\paulitrivial$.
\end{proof}

\begin{cor}\label{cor:pqg_decomposition}
Every element of $\pauli$ admits a unique decomposition ---up to phases--- as $pqg$ with $p\in\pauli_\graph$, $q\in\pauli_1\subset\pauli$, $g\in\gauge_1\subset\gauge-\stab$, where
\be
\coarse L {\pauli_1}:=F[\paulitrivial^{\otimes n}], \qquad\gauge_1:=\cen{\pauli_\graph\pauli_1}\pauli.
\ee

\end{cor}

\subsection{Examples}\label{sec:examples_structure}

Some examples of the construction in \thmref{lattice} follow. 
Only their nontrivial content is given.
For the topological spin extensions we choose $\spin{\tilde c_1}=1$.
 
\begin{itemize}
\item
Toric code and subsystem toric code $(L=1)$
\be
\edgeop 1 {\hedge} = Z_{0,0}^h, \quad \dedgeop 1 {\hedge} = X_{0,0}^h,\quad
 \edgeop 1 {\vedge}= Z_{0,0}^v,\quad \dedgeop 1 {\vedge} = X_{0,0}^v.
\ee

\item
{Fermionic subsystem toric code} 
($L=1$)
\begin{align}
\edgeop 1 {\hedge} = Z_f^h, \quad \edgeop 1 {\hedge^\ast} = X_{0,0}^h,\quad
 \edgeop 1 {\vedge} = Z_f^v,\quad \edgeop 1 {\vedge^\ast} = X_{0,0}^v,
\end{align}

\item
Honeycomb subsystem code ($L=2$)
\begin{align}
\edgeop 1 {\hedge} &=  G^Y_{0,0}G^X_{1,0}G^Y_{1,0}G^X_{2,0},\qquad\edgeop 1 {\hedge^\ast} =Y_{1,0}^1Y_{1,-1}^1,\nl
\edgeop 1 {\vedge} &=  G^Z_{0,0}G^X_{0,1}G^Z_{0,1}G^X_{0,2} ,\qquad\edgeop 1 {\vedge^\ast} = Z_{0,1}^1Z_{-1,0}^2,\nl
\stab_g (1) &= \stab_g^{\hexagon}-\strans 2 (S^{\hexagon}),\qquad \quad g\in\pro {\edgeopi 1 {\dface}},\nl
\gauge_g (1) &=\gauge_g\cup\strans 2 (\sset{g G^Y, g G^Z})-\strans 2 (\sset{G^X,G^Y,G^Z}),\nl
\end{align}

\item
Topological subsystem color code ($L=3$)
\begin{align}
\edgeop 1 {\hedge} &:= Z_{1,0}^2 X_{1,0}^3 Y_{1,0}^1 X_{1,0}^5 X_{2,0}^2 Y_{2,0}^4  X_{2,0}^6 Z_{2,0}^5 Y_{3,0}^4 Y_{3,0}^2 X_{3,0}^6 X_{3,0}^3 Y_{3,0}^5  Y_{3,0}^1,\nl
\edgeop 1 {\vedge} &:=  Z_{0,0}^1 Z_{0,1}^6 Y_{0,1}^5 Y_{0,1}^3 X_{0,1}^1 X_{0,2}^4 Y_{0,2}^6  Y_{0,2}^2 Z_{0,2}^3 X_{0,3}^4 Y_{0,3}^2 X_{0,3}^6 X_{0,3}^3 Y_{0,3}^5 Y_{0,3}^1,\nl
\dedgeop 1 {\hedge} &:=  \trans 2 {-2} (\edgeop 1 {\vedge}),\qquad 
\dedgeop 1 {\vedge} :=  \trans {-1} 1 (\edgeop 1 {\hedge}),\nl
G_g(1) &:= \gaugeTSCCg - \strans 3 (\sset{G^1_{0,0},G^1_{2,1}}),\nl
S_g(1) &:=  \stabTSCCg \cup \strans 3 (\sset{S_{0,0}^1S_{0,0}^2,S_{2,0}^2S_{2,1}^1,S_{2,1}^1S_{2,1}^2,S_{0,2}^2S_{3,0}^1})-\nl 
&- \strans 3 (\sset{S_{0,0}^1,S_{0,0}^2,S_{2,0}^2,S_{2,1}^1,S_{2,1}^2,S_{0,2}^2}),
\end{align}
\end{itemize}

\subsection{Local equivalence}\label{sec:equivalence}

Among the examples above, the three variants of the toric code have a particularly simple structure, suggesting the following definition.
\begin{defn}
A TSSG $\stab$ is elementary if it admits the structure of \thmref{lattice} with $\pauli_\graph=\pauli$.
\end{defn}
From our examples it follows that there exist elementary TSSGs of any non-chiral characteristic.
When $\beta=0$, an elementary TSSG is a TSG.

\begin{prop}\label{prop:elementary}
An elementary TSSG is locally equivalent to any TSSG with the same characteristic.
\end{prop}

\begin{proof}
If $\stab'\subset\pauli'$ is an elementary TSSG, it admits $\edgeopxp k$ as in \thmref{lattice} with $\pauli_\graph=\pauli'$.
Let $\stab\subset\pauli$ be a TSSG with the same characteristic.
Using the same extension of the topological spin as for $\edgeopxp k$, choose $\edgeopx k$ and $\stab_g$ according \thmref{lattice}.
Since we are allowed to coarse grain, w.l.o.g. we assume that $L=1$ in both cases.
We can construct a LPG morphism $G$ from $\stab'\subset\pauli'$ to $\stab\cap\pauli_\graph\subset\pauli$ and with image in $\pauli_\graph$.
Namely, set $G[\edgeopp k e]=\psi(e,k)\sigma(e,k)\edgeop k e$ for any $k\in\indc\cup\inde$ and $e\in\edges$, where $\psi(e,k)=1,i$ and $\sigma(e,k)=\pm 1$ are translationally invariant and as follows.
There is a unique choice of $\psi(k,e)$ to preserve self-adjointness.
We set $\sigma(e,k)=1$ unless $e$ is horizontal and belongs to an even row, and then it is easy to check that there is a unique choice for $\sigma$ that maps stabilizers to stabilizers, without sign flips.
Now choose $\funct F {\paulitrivial^{\otimes n}} \pauli$ according to \propref{stab_injection}.
Consider the LPG morphism $\funct H {\pauli'\otimes\paulitrivial^{\otimes n}}{\pauli}$ from $\stab'\otimes\stabtrivial^{\otimes n}$ to $\stab$ defined by
$H(p\otimes q)=F(p)G(q)$.
According to \lemref{extension} we can extend $H$ to an isomorphism $\funct {\hat H}{\pauli'\otimes\paulitrivial^{\otimes n+m}}{\pauli}$ that gives the local equivalence.
\end{proof}

\begin{cor}\label{cor:equivalence_nc}
Two non-chiral TSSGs are locally equivalent iff they have the same characteristic.
\end{cor}

\begin{prop}\label{thm:equivalence}
Given a TSG, there exists a elementary TSG with the same characteristic.
\end{prop}

\begin{proof}
Let $\stab\subset \pauli$ be a TSG and $\funct F{\paulitrivial^{\otimes n}}{\pauli}$ the LPG isomorphism of \propref{stab_injection}, where
we assume w.l.o.g. $L=1$.
Then $\pauli_1=F[\paulitrivial^{n}]$ and $\gauge_1=\phase$ in \corref{pqg_decomposition}.
According to \lemref{extension} there exists for some $m\in\N_0$ a lattice Pauli morphism $\funct G{\paulitrivial^{m}}{\pauli}$ with $G[\paulitrivial^{m}]=\cen{\pauli_1}\pauli=\pauli_\graph$ and period 1.
The desired elementary TSG is then $\stab_0\subset\pauli_0$, with $\stab_0:=G\inv[\stab\cap\pauli_\graph]\subset\paulitrivial^{m}=:\pauli_0$.
Indeed, $G$ is a monomorphism and $\cen{\pauli_\graph\cap\stab}{\pauli_\graph}\propto\pauli_\graph\cap\stab$.
\end{proof}

\begin{cor}\label{cor:equivalence_TSG}
Two TSGs are locally equivalent iff they have the same characteristic.
\end{cor}

\section{Conclusions and outlook}\label{sec:conclusions}

Under the sole assumption of translational invariance, we have provided a detailed study of the structure of two-dimensional TSCs.
This structure can always be understood in terms of string operators that carry a topological charge, allowing to extend the insights from well known codes such as the toric code.
Codes can be classified in terms of their topological charges, which are invariant under local transformations.
Both for subspace codes and non-chiral subsystem codes, we show that two codes with isomorphic topological charges can be related by a local transformation.
The existence of chiral subspace codes remains open.

From a computational perspective, the results have interesting implications. 
The relevant charges in a code can be arranged into several copies of the toric code or subsystem color codes. 
But by means of code deformations \cite{bombin:2009:deformation}, in these codes the whole Clifford group of gates can be implemented in a topologically protected way. 
In the case of subsystem color codes this only involves the introduction of twist defects to encode qubits \cite{bombin:2010:subsystem}. 
For toric codes, single qubit Clifford gates can be recovered by encoding in twists \cite{bombin:2010:twist} and CNot gates by encoding in hole pairs \cite{raussendorf:2007:deformation,bombin:2009:deformation}. 
But it is possible to switch between both encodings, so that also in the toric code deformations are enough to recover the whole Clifford group. 
As a result ---since such techniques do not depend on details of the codes but rather on the charge content---, we can perform all Clifford gates by code deformation in any two-dimensional TSC.

We finish with a discussion of some natural extensions of the present work.

\begin{sep}{Boundaries}
The same ideas that we have used to classify TSCs can be applied to classify boundaries between them.
To model them, we can take the right and left half-plane to correspond to two possibly different TSSGs, allowing for arbitrary but translationally symmetric ---along the axis direction--- gauge and stabilizer generators in the axis between them.
Of course, we have to impose the topological condition on $\stab$, which still makes sense because we do not want to have localized degrees of freedom on small portions of the boundary.

Clearly there exist natural morphisms from the charge groups of the two TSSGs to the charge group of this boundary TSSG.
These turn out to be enough to label all the charges ---because no global constraints, and thus no charge, can be confined to the axis.
For charges not confined to either half-plane, the values of $\mutualx$ and $\spinx$ on each side must agree.
This is in particular true for the trivial charge.
Thus, those charges that form the kernel of each of the two natural morphisms must be bosons with trivial mutual interactions.
Such charges are especially relevant because they `disolve' in the boundary, giving rise to logical string operators with endpoints on the boundary \cite{bombin:2010:subsystem}.
\end{sep}

\begin{sep}{General codes}
It would be more interesting to investigate the structure of general two-dimensional topological quantum error correcting codes, defined in terms of LTI sets of commuting projectors.
The expected result is that such codes would be describable in terms of anyon models with additional structure.
\end{sep}

\begin{sep}{Higher dimensions}
Even in the case of the relatively simple stabilizer formalism, the general structure of (translationally invariant) topological codes in three dimensions and above turns out to be quite rich.
Already in the three-dimensional case there exist examples that do not fit on the standard homological picture \cite{bravyi:2010:topological,haah:2011:local}.
This is of great interest due to the possible thermal stability \cite{bravyi:2011:energy} of the corresponding quantum memories, protected by the local Hamiltonian.
\end{sep}

\vspace{.5cm}
\noindent{\emph {Acknowlegments}---}
I would like to acknowledge useful discussions with Sergey Bravyi and Daniel Gottesman.
I thank especially Guillaume Duclos-Cianci and David Poulin for many discussions and comments on several versions of the paper.
This work was supported with research grants QUITEMAD S2009-ESP-1594, FIS2009-10061 and UCM-BS/910758.
Work at Perimeter Institute is supported by Industry Canada and Ontario MRI.

 \bibliographystyle{ieeetr}
 \bibliography{refs}

\end{document}